\documentclass[runningheads]{llncs}

\usepackage{hyperref}
\usepackage{amsmath}
\usepackage{amssymb}
\usepackage{graphicx}

\newcommand{\ccNP}{\textrm{\textsc{NP}}}
\newcommand{\row}{\mathcal{R}}
\newcommand{\block}{\mathcal{B}}

\newif\ifabstract
\abstracttrue
\abstractfalse
\newif\iffull
\ifabstract \fullfalse \else \fulltrue \fi

\newtoks\magicAppendix
\magicAppendix={}
\newtoks\magictoks
\newif\iflater
\laterfalse
\ifabstract
\long\def\later#1{\magictoks={#1}%
  \edef\magictodo{\noexpand\magicAppendix={\the\magicAppendix \par
    \noexpand\setcounter{theorem}{\arabic{theorem}}%
    \noexpand\setcounter{lemma}{\arabic{lemma}}%
    \noexpand\setcounter{definition}{\arabic{definition}}%
    \noexpand\setcounter{corollary}{\arabic{corollary}}%
    \the\magictoks}}%
  \magictodo}
\long\def\both#1{\later{#1}\the\magictoks}
\else
\long\def\later#1{#1}
\long\def\both#1{#1}
\fi
\def\magicappendix{\latertrue \the\magicAppendix}

\begin{document}

\title{Staged Self-Assembly and\\Polyomino Context-Free Grammars}

\author{Andrew Winslow\thanks{Supported in part by National Science Foundation grant CBET-0941538.
}}

\authorrunning{A. Winslow}

\institute{
Department of Computer Science, Tufts University, \\
\protect\url{awinslow@cs.tufts.edu}
}

\maketitle

\begin{abstract}
Previous work by Demaine et al. (2012) developed a strong connection between smallest context-free grammars and staged self-assembly systems for one-dimensional strings and assemblies.
We extend this work to two-dimensional polyominoes and assemblies, comparing staged self-assembly systems to a natural generalization of context-free grammars we call \emph{polyomino context-free grammars (PCFGs)}.

We achieve nearly optimal bounds on the largest ratios of the smallest PCFG and staged self-assembly system for a given polyomino with $n$ cells.
For the ratio of PCFGs \emph{over} assembly systems, we show that the smallest PCFG can be an $\Omega(n/\log^3{n})$-factor larger than the smallest staged assembly system, even when restricted to square polyominoes.
For the ratio of assembly systems over PCFGs, we show that the smallest staged assembly system is never more than a $O(\log{n})$-factor larger than the smallest PCFG and is sometimes an $\Omega(\log{n}/\log\log{n})$-factor larger.
\end{abstract}


\section{Introduction}

In the mid-1990s, the Ph.D. thesis of Erik Winfree~\cite{Winfree-1998} introduced a theoretical model of self-assembling nanoparticles.
In this model, which he called the \emph{abstract tile assembly model (aTAM)}, square particles called \emph{tiles} attach edgewise to each other if their edges share a common \emph{glue} and the bond strength is sufficient to overcome the kinetic energy or \emph{temperature} of the system.
The products of these systems are \emph{assemblies}: aggregates of tiles forming via crystal-like growth starting at a \emph{seed tile}.
Surprisingly, these tile systems have been shown to be computationally universal~\cite{Winfree-1998,Cook-2011}, self-simulating~\cite{Doty-2010,Doty-2012a}, and capable of optimally encoding arbitrary shapes~\cite{Rothemund-2000,Adleman-2001,Soloveichik-2005}.

\iffull
In parallel with work on the aTAM, a number of variations on the model have been proposed and investigated.
These variations change a number of features of the original aTAM, for instance allowing glues to repulse~\cite{Doty-2011,Patitz-2011,Schweller-2012}, or adding labels to each tile to produce patterned assemblies~\cite{Goos-2011,Czeizler-2012,Seki-2013}.
For a more thorough treament of the aTAM and its variants, see the recent surveys of Patitz~\cite{Patitz-2012} and Doty~\cite{Doty-2012b}.

\else
In parallel with work on the aTAM, a number of variations on the model have been proposed and investigated.
\fi
One well-studied variant called the \emph{hierarchical}~\cite{Chen-2012} or \emph{two-handed assembly model (2HAM)}~\cite{Demaine-2008} eliminates the seed tile and allows tiles and assemblies to attach in arbitrary order.
This model was shown to be capable of (theoretically) faster assembly of squares~\cite{Chen-2012} and simulation of aTAM systems~\cite{Cannon-2013}, including capturing the seed-originated growth dynamics.
A generalization of the 2HAM model proposed by Demaine et al.~\cite{Demaine-2008} is the \emph{staged assembly model}, which allows the assemblies produced by one system to be used as reagents (in place of tiles) for another system, yielding systems divided into sequential assembly \emph{stages}.
They showed that such sequential assembly systems can replace the role of glues in encoding complex assemblies, allowing the construction of arbitrary shapes efficiently while only using a constant number of glue types, a result impossible in the aTAM or 2HAM.

To understand the power of the staged assembly model, Demaine et al.~\cite{Demaine-2012a} studied the problem of finding the smallest system producing a one-dimensional assembly with a given sequence of labels on its tiles, called a \emph{label string}. 
They proved that for systems with a constant number of glue types, this problem is equivalent to the well-studied problem of finding the smallest context-free grammar whose language is the given label string, also called the \emph{smallest grammar problem} (see~\cite{Lehman-2002,Charikar-2005}).
For systems with unlimited glue types, they proved that the ratio of the smallest context-free grammar \emph{over} the smallest system producing an assembly with a given label string of length $n$ (which they call \emph{separation}) is $\Omega(\sqrt{n/\log{n}})$ and $O((n/\log{n})^{2/3})$ in the worst case.

In this paper we consider the two-dimensional version of this problem: finding the smallest staged assembly system producing an assembly with a given \emph{label polyomino}.
For systems with constant glue types and no cooperative bonding, we achieve separation of grammars \emph{over} these systems of $\Omega(n/(\log\log{n})^2)$ for polyominoes with $n$ cells (Sect.~\ref{sec:sas-much-better-pcfgs-gen}), and $\Omega(n/\log^3{n})$ when restricted to rectangular (Sect.~\ref{sec:sas-much-better-pcfgs-rectangle}) or square (Sect.~\ref{sec:sas-much-better-pcfgs-square}) polyominoes with a constant number of labels.
Adding the restriction that each step of the assembly process produces a single product, we achieve $\Omega(n/\log^3{n})$ separation for general polyominoes with a single label (Sect.~\ref{sec:sas-much-better-pcfgs-gen}). 
For the separation of staged assembly systems \emph{over} grammars, we achieve bounds of $\Omega(\log{n}/\log\log{n})$ (Sect.~\ref{sec:pcfgs-slightly-better-sas-ssas}) and, constructively, $O(\log{n})$ (Sect.~\ref{sec:pcfgs-never-much-better-sas}).
For all of these results, we use a simple definition of context-free grammars on polyominoes that generalizes the deterministic context-free grammars (called \emph{RCFGs}) of~\cite{Demaine-2012a}. 

When taken together, these results give a nearly complete picture of how smallest context-free grammars and staged assembly systems compare.
For some polyominoes, staged assembly systems are exponentially smaller than context-free grammars ($O(\log{n})$ vs. $\Omega(n/\log^3{n})$).
On the other hand, given a polyomino and grammar deriving it, one can construct a staged assembly system that is a (nearly optimal) $O(\log{n})$-factor larger and produces an assembly with a label polyomino replicating the polyomino.

 

\section{Staged Self-Assembly}
\label{sec:sas-defns}

An instance of the staged tile assembly model is called a \emph{staged assembly system} or \emph{system}, abbreviated \emph{SAS}.
A SAS $\mathcal{S} = (T, G, \tau, M, B)$ is specified by five parts:
a \emph{tile set} $T$ of square \emph{tiles},
a \emph{glue function} $G : \Sigma(G)^2 \rightarrow \{0, 1, \dots, \tau \}$,
a \emph{temperature} $\tau \in \mathbb{N}$,
a directed acyclic \emph{mix graph} $M = (V, E)$,
and a \emph{start bin function} $B : V_L \rightarrow T$ from the \emph{leaf vertices} $V_L \subseteq V$ of $M$ with no incoming edges. 

Each tile $t \in T$ is specified by a 5-tuple $(l, g_n, g_e, g_s, g_w)$ consisting of a label $l$ taken from an alphabet $\Sigma(T)$ (denoted $l(t)$) and a set of four non-negative integers in $\Sigma(G) = \{0, 1, \dots, k\}$ specifying the \emph{glues} on the sides of $t$ with normal vectors $\langle 0, 1 \rangle$ (north), $\langle 1, 0 \rangle$ (east), $\langle 0, -1, \rangle$ (south), and $\langle -1, 0 \rangle$ (west), respectively, denoted $g_{\vec{u}}(t)$.
In this work we only consider glue functions with the constraints that if $G(g_i, g_j) > 0$ then $g_i = g_j$, and $G(0, 0) = 0$.

A \emph{configuration} is a partial function $C : \mathbb{Z}^2 \rightarrow T$ mapping locations on the integer lattice to tiles.
Any two locations $p_1 = (x_1, y_1)$, $p_2 = (x_2, y_2)$ in the domain of $C$ (denoted ${\rm dom}(C)$) are \emph{adjacent} if $||p_2 - p_1|| = 1$ and the \emph{bond strength} between any pair of tiles $C(p_1)$ and $C(p_2)$ at adjacent locations is $G(g_{p_2 - p_1}(C(p_1)), g_{p_1 - p_2}(C(p_2))$.
A \emph{configuration} is a \emph{$\tau$-stable assembly} or an \emph{assembly at temperature $\tau$} if ${\rm dom}(C)$ is connected on the lattice and, for any partition of ${\rm dom}(C)$ into two subconfigurations $C_1$, $C_2$, the sum of the bond strengths between tiles at pairs of locations $p_1 \in {\rm dom}(C_1)$, $p_2 \in {\rm dom}(C_2)$ is at least $\tau$. 
Any pair of configurations $C_1$, $C_2$ are equivalent if there exists a vector $\vec{v} = \langle x, y \rangle$ such that ${\rm dom}(C_1) = \{ p + \vec{v} \mid p \in {\rm dom}(C_2) \}$ and $C_1(p) = C_2(p + \vec{v})$ for all $p \in {\rm dom}(C_1)$.
Two $\tau$-stable assemblies $A_1$, $A_2$ are said to \emph{assemble} into a \emph{superassembly} $A_3$ if there exists a translation vector $\vec{v} = \langle x, y \rangle$ such that ${\rm dom}(A_1) \cap \{ p + \vec{v} \mid p \in A_2 \} = \emptyset$ and $A_3$ defined by the partial functions $A_1$ and $A_2'$ with $A_2'(p) = A_2(p + \vec{v})$ is a $\tau$-stable assembly. 

Each vertex of the mix graph $M$ describes a \emph{two-handed assembly process}.
This process starts with a set of $\tau$-stable \emph{input assemblies} $I$.
The set of \emph{assembled assemblies} $Q$ is defined recursively as $I \subseteq Q$, and for any pair of assemblies $A_1, A_2 \in Q$ with superassembly $A_3$, $A_3 \in Q$.
Finally, the set of \emph{products} $P \subseteq Q$ is the set of assemblies $A$ such that for any assembly $A'$, no superassembly of $A$ and $A'$ is in $Q$.

The mix graph $M = (V, E)$ of $\mathcal{S}$ defines a set of two-handed assembly processes (called \emph{mixings}) for the non-leaf vertices of $M$ (called \emph{bins}).
The input assemblies of the mixing at vertex $v$ is the union of all products of mixings at vertices $v'$ with $(v', v) \in E$.
The start bin function $B$ defines the lone single-tile product of each mixings at a leaf bin. 
The system $\mathcal{S}$ is said to \emph{produce} an assembly $A$ if some mixing of $\mathcal{S}$ has a single product, $A$. 
We define the size of $\mathcal{S}$, denoted $\mathcal{S}$, to be $|E|$, the number of edges in $M$.
If every mixing in a $\mathcal{S}$ has a single product, then $\mathcal{S}$ is a \emph{singular self-assembly system (SSAS)}.
 
\begin{figure}[ht]
\centering
\includegraphics[scale=1.0]{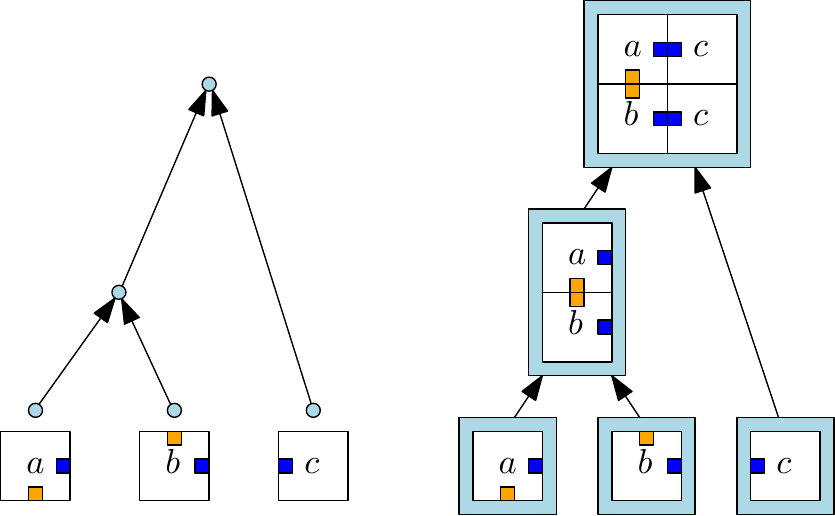}
\caption{A self-assembly system (SAS) consisting of a mix graph and tile types (left), and the assemblies produced by carrying out the algorithmic process of staged self-assembly (right).}
\label{fig:2D-sas-md}
\end{figure}

The results of Section~\ref{sec:constant-glues} use the notion of a self-assembly system $\mathcal{S}'$ \emph{simulating} a system $\mathcal{S}$ by carrying out the same sequence of mixings and producing a set of scaled assemblies.
Formally, we say a system $\mathcal{S}' = (T', G', \tau, M', B')$ \emph{simulates} a system $\mathcal{S} = (T, G, \tau, M, B)$ at \emph{scale} $b$ if there exist two functions $f$, $g$ with the following properties:

\begin{enumerate}
\item[(1)] The function $f : (\Sigma(T') \cup \{\varnothing\}) ^{b^2} \rightarrow \Sigma(T) \cup \{\varnothing\}$ maps the labels of $b \times b$ regions of tiles (called \emph{blocks}) to a label of a tile in $T$.
The empty label $\varnothing$ denotes no tile.
\item[(2)] The function $g : S' \rightarrow V$ maps a subset $S'$ of the vertices of the mix graph $M'$ to vertices of the mix graph $M$ such that $g$ is an isomorphism between the subgraph induced by $S'$ in $M'$ and the graph $M$.
\item[(3)] Let $P(v)$ be the set of products of the bin corresponding to vertex $v$ in a mix graph.
Then for each vertex $v \in M$ with $v' = g^{-1}(v)$, $P(v) = \{f(p) \mid p \in P(v')\}$.
\end{enumerate}

\iffull
Intuitively, $f$ defines a correspondence between the $b$-scaled macrotiles in $\mathcal{S}'$ simulating tiles in $\mathcal{S}$, and $g$ defines a correspondence between bins in the systems.
Property (3) requires that $f$ and $g$ do, in fact, define correspondence between what the systems produce.
\fi

The self-assembly systems constructed in Sections~\ref{sec:pcfgs-never-much-better-sas} and~\ref{sec:sass-much-better-than-pcfgs} produce only \emph{mismatch-free assemblies}: assemblies in which every pair of incident sides of two tiles in the assembly have the same glue.
A system is defined to be \emph{mismatch-free} if every product of the system is mismatch-free.

\section{Polyomino Context-Free Grammars}

Here we describe polyominoes, a generalization of strings, and polyomino context-free grammars, a generalization of deterministic context-free grammars.
These objects replace the strings and restricted context-free grammars (RCFGs) of Demaine et al.~\cite{Demaine-2012a}.

A \emph{labeled polyomino} or \emph{polyomino} $P = (S, L)$ is defined by a connected set of points $S$ on the square lattice (called \emph{cells}) containing $(0, 0)$ and a label function $L: S \rightarrow \Sigma(P)$ mapping each cell of $P$ to a \emph{label} contained in an alphabet $\Sigma(P)$.
The \emph{size} of $P$ is the number of cells $P$ contains and is denoted $|P|$.
The label of the cell at lattice point $(x, y)$ is denoted $L((x, y))$ and we define $P(x, y) = L((x, y))$ for notational convenience.
We refer to the \emph{label} or \emph{color} of a cell interchangeably.

Define a \emph{polyomino context-free grammar (PCFG)} to be a quadruple $G = (\Sigma, \Gamma, S, \Delta)$.
The set $\Sigma$ is a set of \emph{terminal symbols} and the set $\Gamma$ is a set of \emph{non-terminal symbols}.
The symbol $S \in \Gamma$ is a special \emph{start symbol}.
Finally, the set $\Delta$ consists of \emph{production rules}, each of the form $N \rightarrow (R_1, (x_1, y_1)) \dots (R_j, (x_j, y_j))$ where $N \in \Gamma$ and is the left-hand side symbol of only this rule, $R_i \in N \cup T$, and each $(x_i, y_i)$ is a pair of integers.
The \emph{size} of $G$ is defined to be the total number of symbols on the right-hand sides of the rules of $\Delta$.

A polyomino $P$ can be derived by starting with $S$, the start symbol of $G$, and repeatedly replacing a non-terminal symbol with a set of non-terminal and terminal symbols.
The set of valid replacements is $\Delta$, the production rules of $G$, where a non-terminal symbol $N$ with lower-leftmost cell at $(x, y)$ can be replaced with a set of symbols $R_1$ at $(x + x_1, y + y_1)$, $R_2$ at $(x + x_2, y + y_2)$, $\dots$, $R_j$ at $(x + x_j, y + y_j)$ if there exists a rule $N \rightarrow (R_1, (x_1, y_1)) (R_2, (x_2, y_2)) \dots (R_j, (x_j, y_j))$.
Additionally, the set of terminal symbol cells derivable starting with $S$ must be connected and pairwise disjoint.

The polyomino $P$ derived by the start symbol of a grammar $G$ is called the \emph{language of $G$}, denoted $L(G)$, and $G$ is said to \emph{derive} $P$.
In the remainder of the paper we assume that each production rule has at most two right-hand side symbols (equivalent to binary normal form for 1D CFGs), as any PCFG can be converted to this form with only a factor-2 increase in size.
Such a conversion is done by iteratively replacing two right-hand side symbols $R_i$, $R_{i'}$ with a new non-terminal symbol $Q$, and adding a new rule replacing $Q$ with $R_i$ and $R_{i'}$. 

Intuitively, a polyomino context-free grammar is a recursive decomposition of a polyomino into smaller polyominoes.
Because each non-terminal symbol is the left-hand side symbol of at most one rule, each non-terminal corresponds to a subpolyomino of the derived polyomino.
Then each production rule is a decomposition of a subpolyomino into smaller subpolyominoes (see Figure~\ref{fig:pcfg-prod-rule-explanation}).

\begin{figure}[ht]
\centering
\includegraphics[scale=1.0]{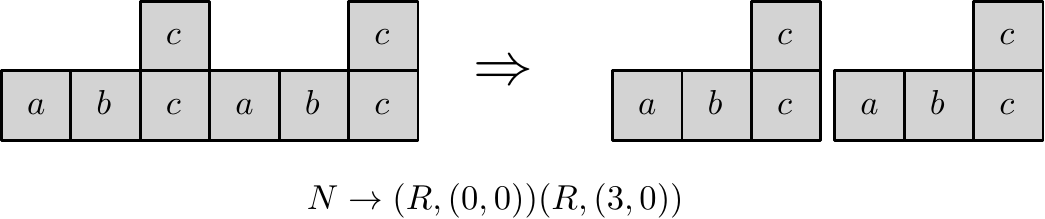}
\caption{Each production rule in a PCFG generating a single shape is a decomposition of the left-hand side non-terminal symbol's polyomino into the right-hand side symbols' polyominoes.}
\label{fig:pcfg-prod-rule-explanation}
\end{figure}

In this interpretation, the smallest grammar deriving a given polyomino is equivalent to a decomposition using the fewest \emph{distinct} subpolyominoes in the decomposition.
As for the smallest CFG for a given string, the smallest PCFG for a given polyomino is deterministic and finding such a grammar is \ccNP-hard.
Moreover, even approximating the smallest grammar is \ccNP-hard~\cite{Charikar-2005}, and achieving optimal approximation algorithms remains open~\cite{Jez-2013}.

In Section~\ref{sec:pcfgs-never-much-better-sas} we construct self-assembly systems that produce assemblies whose label polyominoes are scaled versions of other polyominoes, with some amount of ``fuzz'' in each scaled cell.
A polyomino $P' = (S', L')$ is said to be a \emph{$(c, d)$-fuzzy replica} of a polyomino $P = (S, L)$ if there exists a vector $\langle x_t, y_t \rangle$ with the following properties:

\begin{enumerate}
\item For each block of cells $\mathcal{S}'_{(i, j)} = \{ (x, y) \mid x_t + di \leq x < x_t + d(i+1), y_t + dj \leq y < y_t + d(j+1)\}$ (called a \emph{supercell}), $\mathcal{S}'_{(i, j)} \cap S' \neq \varnothing$ if and only if $(i, j) \subseteq S$.
\item For each supercell $\mathcal{S}'_{(i, j)}$ containing a cell of $P'$, the subset of \emph{label cells} $\{ (x, y) \mid x_t + di + (d-c)/2 \leq x < x_t + d(i+1) + (d-c)/2, y_t + dj + (d-c)/2 \leq y < y_t + d(j+1) + (d-c)/2\}$ consists of $c^2$ cells of $P'$, with all cells having identical label, called the \emph{label of the supercell} and denoted $\mathcal{L}_{(i, j)}$.
\item For each supercell $\mathcal{S}'_{(i, j)}$, any cell that is not a label cell of $\mathcal{S}'_{(i, j)}$ has a common \emph{fuzz label} in $L'$.
\item For each supercell $\mathcal{S}'_{(i, j)}$, the label of the supercell $\mathcal{L}'_{(i, j)} = P(i, j)$.
\end{enumerate}

Properties~(1) and~(2) define how sets of cells in $P'$ replicate individual cells in $P$, and the labels of these sets of cells and individual cells.
Property~(3) restricts the region of each supercell not in the label region to contain only cells with a common fuzz label.
Property~(4) requires that each supercell's label matches the label of the corresponding cell in $P$.

\section{SAS over PCFG Separation Lower Bound}
\label{sec:pcfgs-slightly-better-sas-ssas}

This result uses a set of shapes we call \emph{$n$-stagglers}, an example is seen in Figure~\ref{fig:staggler}.
The shapes consist of $\log{n}$ bars of dimensions $n/\log{n} \times 1$ stacked vertically atop each other, with each bar horizontally offset from the bar below it by some amount in the range $-(n/\log{n}-1), \dots, n/\log{n}-1$.
We use the shorthand that $\log{n} = \lfloor \log{n} \rfloor$ for conciseness. 
Every sequence of $\log{n}-1$ integers, each in the range $[-(n/\log{n}-1), n/\log{n}-1]$, encodes a unique staggler and by the pidgeonhole principle, some $n$-staggler requires $\log((2n/\log{n}-1)^{\log{n}-1} = \Omega(\log^2{n})$ bits to specify.

\begin{figure}[ht]
\centering
\includegraphics[scale=1.0]{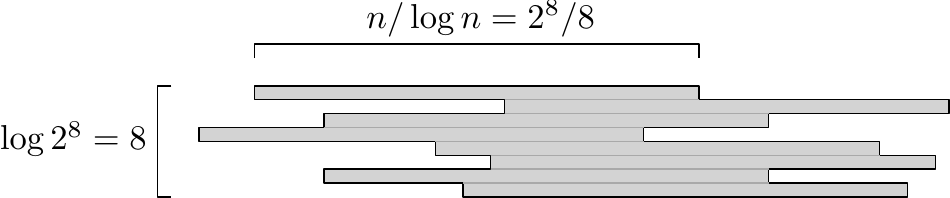}
\caption{The $2^8$-staggler specified by the sequence $-18, 13, 9, -17, -4, 12, -10$.}
\label{fig:staggler}
\end{figure}

\begin{lemma}
Any $n$-staggler can be derived by a PCFG of size $O(\log{n})$. 
\end{lemma}

\begin{proof}
A set of $O(\log{n})$ production rules deriving a bar (of size $\Theta(n/\log{n}) \times 1$) can be constructed by repeatedly doubling the length of the bar, using an additional $\log{n}$ rules to form the bar's exact length.
The result of these production rules is a single non-terminal $B$ deriving a complete bar.

Using the non-terminal $B$, a stack of $k$ bars can be described using a production rule $N \rightarrow (B, (x_1, 0)) (B, (x_2, 1)) \dots (B, (x_k, k-1))$, where the x-coordinates $x_1, x_2, \dots, x_k$ encode the offsets of each bar relative to the bar below it.
An equivalent set of $k-1$ production rules in binary normal form can be produced by creating a distinct non-terminal for $T_i$ each stack of the first $i$ bars, and a production rule $T_i \rightarrow (T_{i-1}, (0, 0)) (B, (x_i, i))$ encoding the offset of the topmost bar relative to the stack of bars beneath it.  

In total, $O(\log{n})$ rules are used to create $B$, the non-terminal deriving a bar, and $O(\log{n})$ are used to create the stack of bars, one per bar.
So the $n$-staggler can be constructed using a PCFG of size $O(\log{n})$.
\end{proof}

\begin{lemma}
For every $n$, there exists an $n$-staggler $P$ such that any SAS or SSAS producing an assembly with label polyomino $P$ has size $\Omega(\log^2{n}/\log\log{n})$. 
\end{lemma}

\begin{proof}
The proof is information-theoretic.
Recall that more than half of all $n$-stagglers require $\Omega(\log^2{n})$ bits to specify. 
Now consider the number of bits contained in a SAS $\mathcal{S}$.
Recall that $|\mathcal{S}|$ is the number of edges in the mix graph of $\mathcal{S}$.
Any SAS can be encoded naively using $O(|\mathcal{S}|\log{|\mathcal{S}|})$ bits to specify the mix graph, $O(|T|\log{|T|})$ bits to specify the tile set, and $O(|\mathcal{S}|\log{|T|})$ bits to specify the tile type at each leaf node of the mix graph.
Because the number of tile types cannot exceed the size of the mix graph, $|T| \leq |\mathcal{S}|$.
So the total number of bits needed to specify $\mathcal{S}$ (and thus the number of bits of information contained in $\mathcal{S}$) is $O(|\mathcal{S}|\log{|\mathcal{S}|} + |T|\log{|T|} + |\mathcal{S}|\log{|\mathcal{S}|}) = O(|\mathcal{S}|\log{|\mathcal{S}|})$. 
So some $n$-staggler requires a SAS $\mathcal{S}$ such that $O(|\mathcal{S}|\log{|\mathcal{S}|}) = \Omega(\log^2{n})$ and thus $|\mathcal{S}| = \Omega(\log^2{n}/\log{\log{n}})$.
\end{proof}

\begin{theorem}
The separation of SASs and SSASs over PCFGs is $\Omega(\log{n}/\log\log{n})$.
\end{theorem}

\begin{proof}
By the previous two lemmas, more than half of all $n$-stagglers require SASs and SSASs of size $\Omega(\log^2{n}/\log\log{n})$ and all $n$-stagglers have PCFGs of size $O(\log{n})$. 
So the separation is $\Omega(\log{n}/\log\log{n})$.
\end{proof}

We also note that scaling the $n$-staggler by a $c$-factor produces a shape which is derivable by a CFG of size $O(\log{n} + \log{c})$.
That is, the result still holds for $n$-stagglers scaled by any amount polynomial in $n$.
For instance, the $O(n)$-factor of the construction of Theorem~\ref{thm:sas-over-pcfg-upper-bound}.

At first it may not be clear how PCFGs achieve smaller encodings.
After all, each rule in a PCFG $G$ or mixing in SAS $\mathcal{S}$ specifies either a set of right-hand side symbols or set of input bins to use and so has up to $O(\log{|G|})$ or $O(\log{|\mathcal{S}|})$ bits of information.
The key is the coordinate describing the location of each right-hand side symbol.
These offsets have up to $O(\log{n})$ bits of information and in the case that $G$ is small, say $O(\log{n})$, each rule has a number of bits \emph{linear} in the size of the PCFG!

\section{SAS over PCFG Separation Upper Bound}
\label{sec:pcfgs-never-much-better-sas}

Next we show that the separation lower bound of the last section is nearly large as possible by giving an algorithm for converting any PCFG $G$ into a SSAS $\mathcal{S}$ with system size $O(|G|\log{n})$ such that $\mathcal{S}$ produces an assembly that is a fuzzy replica of the polyomino derived by $G$.
Before describing the full construction, we present approaches for efficiently constructing general binary counters and for simulating glues using geometry.

\begin{figure}[h!]
\centering
\includegraphics[scale=1.0]{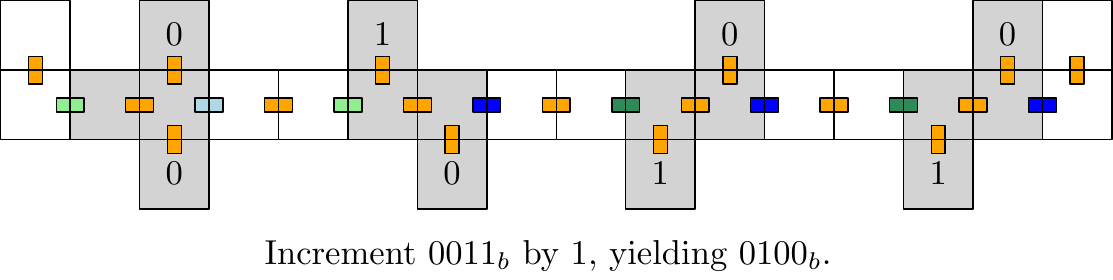}
\caption{A binary counter row constructed using single-bit constant-sized assemblies.
Dark blue and green glues indicate 1-valued carry bits, light blue and green glues indicate 0-valued carry bits.}
\label{fig:binary-counter-row-ex}
\end{figure}

The \emph{binary counter row assemblies} used here are a generalization of those by Demaine et al.~\cite{Demaine-2008} consisting of constant-sized bit assemblies, and an example is seen in Figure~\ref{fig:binary-counter-row-ex}.
Our construction achieves $O(\log{n})$ construction of arbitrary ranges of rows and increment values, in contrast to the contruction of~\cite{Demaine-2008} that only produces row sets of the form $0, 1, \dots, 2^{2^m}-1$ that increment by 1. 
To do so, we show how to construct two special cases from which the generalization follows easily. 

\both{
\begin{lemma}
\label{lem:efficient-incrementors-range-sas}
Let $i,j,n$ be integers such that $0 \leq i \leq j < n$.
There exists a SSAS of size $O(\log{n})$ with a set of bins that, when mixed, assemble a set of $j-i+1$ binary counter rows with values $i, i+1, \dots, j$ incremented by 1.
\end{lemma}
}

\later{
\begin{proof}
Representing integers as binary strings, consider the prefix tree induced by the binary string representations of the integers $i$ through $j$, which we denote $T_{(i, j)}$.
The prefix tree $T_{(0, 2^m-1)}$ is a complete tree of height $m$, and the prefix tree $T_{(i, j)}$ with $0 \leq i \leq j \leq 2^m-1$ is a subtree of $T_{(0, 2^m-1)}$ with $j-i+1$ leaf nodes 
See Figure~\ref{fig:conversion-algo-incrementors} for an example with $m = 4$.

Now let $n = 2^m-1$.
If $T_{(0, n)}$ is drawn with leaves in left-to-right order by increasing integer values, then the leaves of the subtree $T_{(i, j)}$ appear contiguously.
So the subtree $T_{(i, j)}$ has at most $2\log{n}$ internal nodes with one child forming the leftmost and rightmost paths in $T_{(i, j)}$.
Furthermore, if one removes these two paths from $T_{(i, j)}$, the remainder of $T_{(i, j)}$ is a forest of complete trees with at most two trees of each height and $2\log{n}$ trees total.

\begin{figure}[ht]
\centering
\includegraphics[scale=0.8]{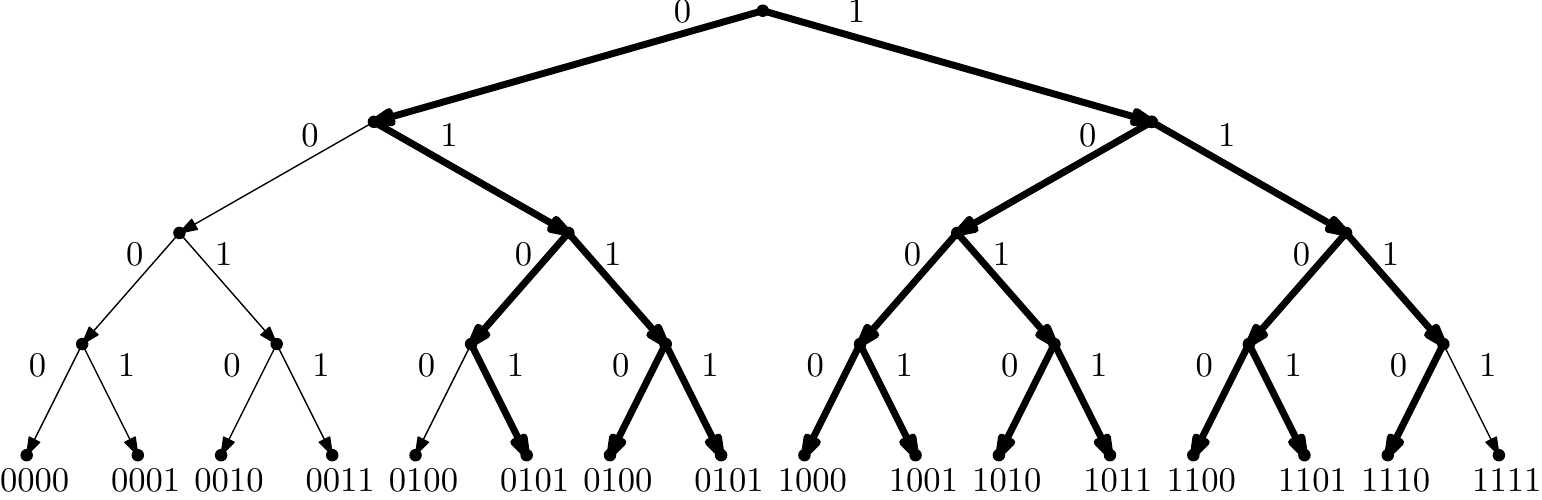}
\caption{The prefix tree $T_{(0, 15)}$ for integers $0$ to $2^4-1$ represented in binary. 
The bold subtree is the prefix subtree $T_{(5, 14)}$ for integers $5$ to $14$.}
\label{fig:conversion-algo-incrementors}
\end{figure}

Note that a complete subtree of the prefix tree corresponds to a set of all possible $2^h$ suffixes of length $h$, where $h$ is the height of the subtree.
The leaves of such a subtree then correspond to the set of strings of length $l$ with a specific prefix of length $l - h$ and any suffix of length $h$.
For the assemblies we use the same geometry-based encoding of each bit as~\cite{Demaine-2008}, and a distinct set of glues used for each bit of the assembly encoding both the bit index and carry bit value from the previous bit.
 
\textsc{Left and right bins}. 
We build a mix graph (seen in Figure~\ref{fig:efficient-counter-range}) consisting of two disjoint paths of bins (called \emph{left bins} and \emph{right bins}) that are used to iteratively assemble partial counter rows $i$ and $j$ by the addition of distinct constant-sized assemblies for each bit.
The partial rows are used to produce the assemblies in the subtree and missing bit bins (described next).
In the suffix trees, the bit strings of these assemblies are progressively longer subpaths of the leftmost and rightmost paths in the subtree of binary strings of the integers $i$ to $j$.

\textsc{Subtree bins}. 
Assemblies in subtree bins correspond to assemblies encoding prefixes of binary counter row values. 
However, unlike left and right bins that encode prefixes of only a single value, subtree bins encode prefixes of many binary counter values between $i$ the $j$ -- namely a set of values forming a maximal complete subtree of the subtree of binary strings of integers from $i$ to $j$, hence the name \emph{subtree bins}.
For example, if $i=12$ and $j=16$, then the set of binary strings for values $12$ ($01100_b$) to $15$ ($01111_b$) have a common prefix $011_b$. 
In this case a subtree bin containing an assembly encoding the three bits $011$ would be created.
Since there are at most $2\log{n}$ such complete subtrees, the number of subtree bins is at most this many.
Creating each bin only requires a single mixing step of combining an assembly from a left or right bin with a single bit assembly, for example adding a $1$-bit assembly to the left bin assembly encoding the prefix $01_b$.

\textsc{Missing bit bins}. 
To add the bits not encoded by the assemblies in the subtree bins, we create sets of four constant-sized assemblies in individual \emph{missing bit bins}.
Since the assemblies in subtree bins encode bit string prefixes of sets of values forming complete subtrees, completing these prefixes with \emph{any} suffix forms a bit string whose value is between $i$ and $j$.
This allows complete non-determinism in the bit assemblies that attach to complete the counter row, provided they properly handle carry bits.
For every bit index missing in \emph{some} subtree bin assembly, the four assemblies encoding the four possibilities for the input and carry values are assembled and placed into separate bins.
When all bins are mixed, subtree assemblies mix non-deterministically with all possible assemblies from missing bit bins, producing all counter rows whose binary strings are found in the subtree.
In total, up to $4\log{n}$ missing bit bins are created, and each contains a constant-sized assembly and so requires constant work to produce.

\begin{figure}[ht]
\centering
\includegraphics[scale=0.8]{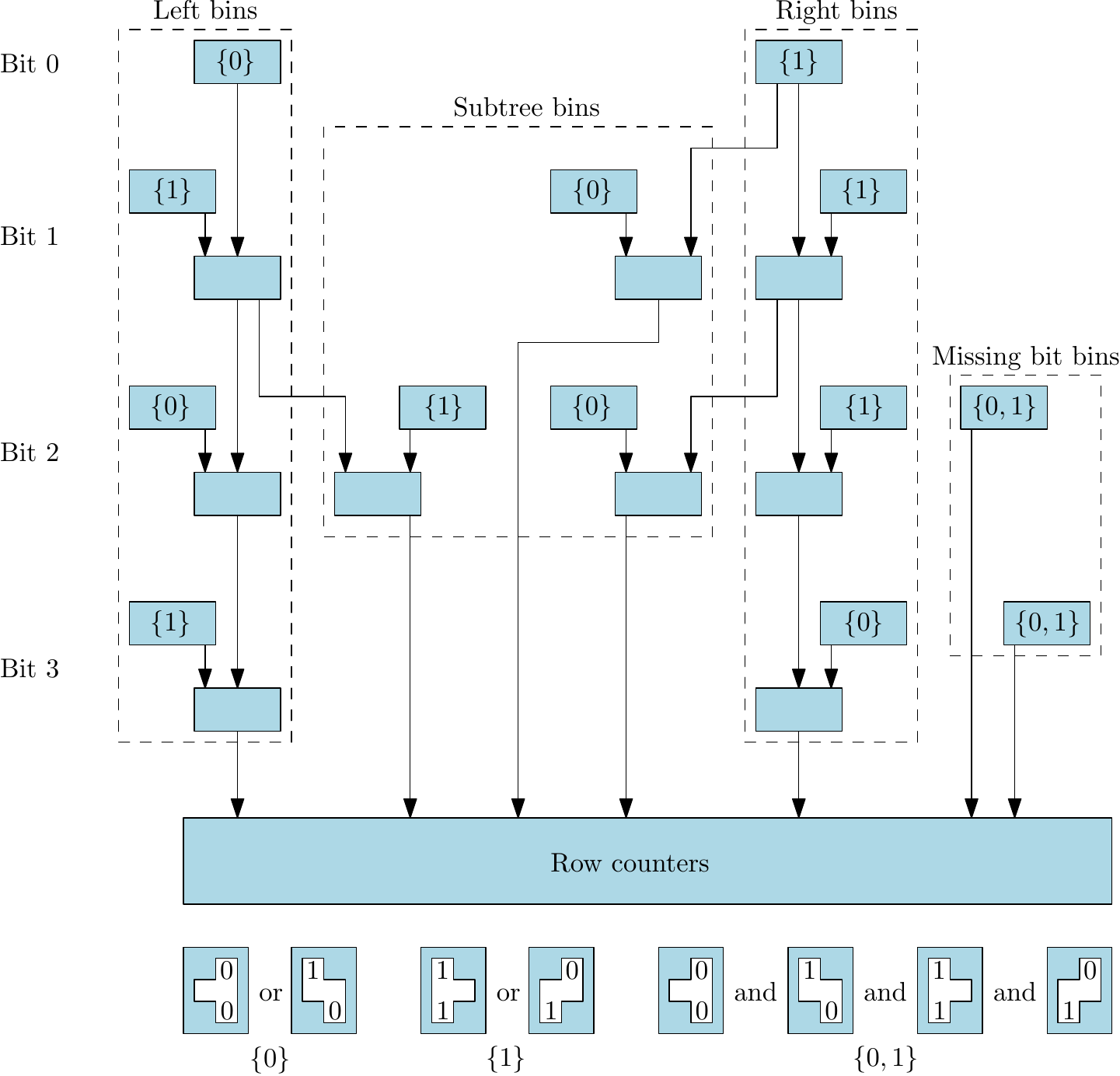}
\caption{The mix graph constructed for the prefix subtree $T_{(5, 14)}$ seen in Figure~\ref{fig:conversion-algo-incrementors}.}
\label{fig:efficient-counter-range}
\end{figure}

The total number of total bins is clearly $O(\log{n})$.
Consider mixing the left and right bins containing completed counter rows for $i$ and $j$, all subtree bins, and all missing bit bins.
Any assembly produced by the system must be a complete binary counter row, as all assemblies are either already complete rows (left and right bins) or are partial assemblies (subtree bins and missing bit bins) that can be extended towards the end of the bit string by missing bit bin assemblies, or towards the start of the bit string by missing bit and then subtree bin assemblies. 
\end{proof}
}

\iffull
The second counter generalization is incrementing by non-unitary values:
\fi

\both{
\begin{lemma}
\label{lem:efficient-incrementors-value-sas}
Let $k, n$ be integers such that $0 \leq k \leq n$ and $n = 2^m$.
There exists a SSAS of size $O(\log{n})$ with a set of bins that, when mixed, assemble a set of $2^m$ binary counter rows with values $0, 1, \dots, 2^m-1$ incremented by $k$.
\end{lemma}
}

\later{
\begin{proof}
For each row, the incremented value of the $b$th bit of the row depends on three values: the previous value of the $b$th bit, the carry bit from the $(b-1)$st addition, and the $b$th bit of $k$.
The resulting output is a pair of bits: the resulting value of the $b$th bit and the $b$th carry bit (seen in Table~\ref{table:incrementor_bit_values}).

\begin{table}
\begin{center}
\begin{tabular}{| c | c | c | c | c |}
\hline
\multicolumn{3}{|c|}{Input bits} & \multicolumn{2}{|c|}{Output bits} \\
\hline
$b$th bit of $k$ & $b$th bit & $(b-1)$st carry bit & $b$th bit & $b$th carry bit \\
\hline
0 & 0 & 0 & 0 & 0 \\
\hline
0 & 0 & 1 & 1 & 0 \\
\hline
0 & 1 & 0 & 1 & 0 \\
\hline
0 & 1 & 1 & 0 & 1 \\
\hline
1 & 0 & 0 & 1 & 0 \\
\hline
1 & 0 & 1 & 0 & 1 \\
\hline
1 & 1 & 0 & 0 & 1 \\
\hline
1 & 1 & 1 & 1 & 1 \\
\hline
\end{tabular}
\end{center}
\caption{All bit combinations for a binary adder incrementing $n$ by $k$.}
\label{table:incrementor_bit_values}
\end{table}
 
Create a set of four $O(1)$-tile subassemblies for each bit of the counter, selecting from the first or second half of the combinations in Table~\ref{table:incrementor_bit_values}, resulting in $4\log{n}$ assemblies total.
Each subassembly handles a distinct combination of the $b$th bit value of the previous row, $(b-1)$st carry bit, and $b$th bit value of $k$ by encoding each possibility as a distinct glue.
When mixed in a single bin, these subassemblies combine in all possible combinations and producing all counter rows from $0$ to $2^m-1$. 
\end{proof}
}

\begin{lemma}
\label{lem:efficient-incrementors-sas}
Let $i,j,k,n$ be integers such that $0 \leq i \leq j < n$ and $0 \leq k \leq n$.
There exists a SSAS of size $O(\log{n})$ with a set of bins that, when mixed, assemble a set of $j-i+1$ binary counter rows with values $i, i+1, \dots, j$ incremented by $k$.
\end{lemma}

\begin{proof}
Combine the constructions used in the proofs of Lemmas~\ref{lem:efficient-incrementors-range-sas} and~\ref{lem:efficient-incrementors-value-sas} by using mixing sequences as in the proof of Lemma~\ref{lem:efficient-incrementors-range-sas} and sets of four subassemblies encoding input, carry, and increment bit values as in the proof of Lemma~\ref{lem:efficient-incrementors-value-sas}.
\end{proof}

Theorem 8 of Demaine et al.~\cite{Demaine-2008} describes how to reduce the number of glues used in a system by replacing each tile with a large \emph{macrotile} assembly, and encoding the tile's glues via unique geometry on the macrotile's sides.
We prove a similar result for labeled tiles, used for proving Theorems~\ref{thm:sas-over-pcfg-upper-bound},~\ref{thm:sas-over-pcfg-upper-bound-constant-glues}, and~\ref{thm:constant-glues}.

\both{
\begin{lemma}
\label{lem:constant-glues}
Any mismatch-free $\tau = 1$ SAS (or SSAS) $\mathcal{S} = (T, G, \tau, M)$ can be simulated by a SAS (or SSAS) $\mathcal{S}'$ at $\tau=1$ with $O(1)$ glues, system size $O(\Sigma(T)|T| + |\mathcal{S}|)$, and $O(\log{|G|})$ scale.
\end{lemma}
}

\later{
\begin{proof}
The proof is constructive.
Produce a set of north \emph{macroglue assemblies} for the glue set: $O(\log{|G|}) \times O(1)$ assemblies, each encoding the integer label of a glue $i$ via a sequence of bumps and dents along the north side of the assembly representing the binary sequence of bits for $i$, as seen in Figure~\ref{fig:glue-bit-encoding}.
All north macroglue assemblies share a pair of common glues: an \emph{inner glue} on the west end of the south side of the assembly (green in Figure~\ref{fig:glue-bit-encoding}) and an \emph{outer glue} on the west end of the north side of the assembly (blue in Figure~\ref{fig:glue-bit-encoding}).
The null glue also has the sequence of bumps and dents (encoding 0), but lacking the outer glue. 
Repeating this process three more times yields sets of east, west, and south macroglue assemblies.

\begin{figure}[ht!]
\centering
\includegraphics[scale=1.0]{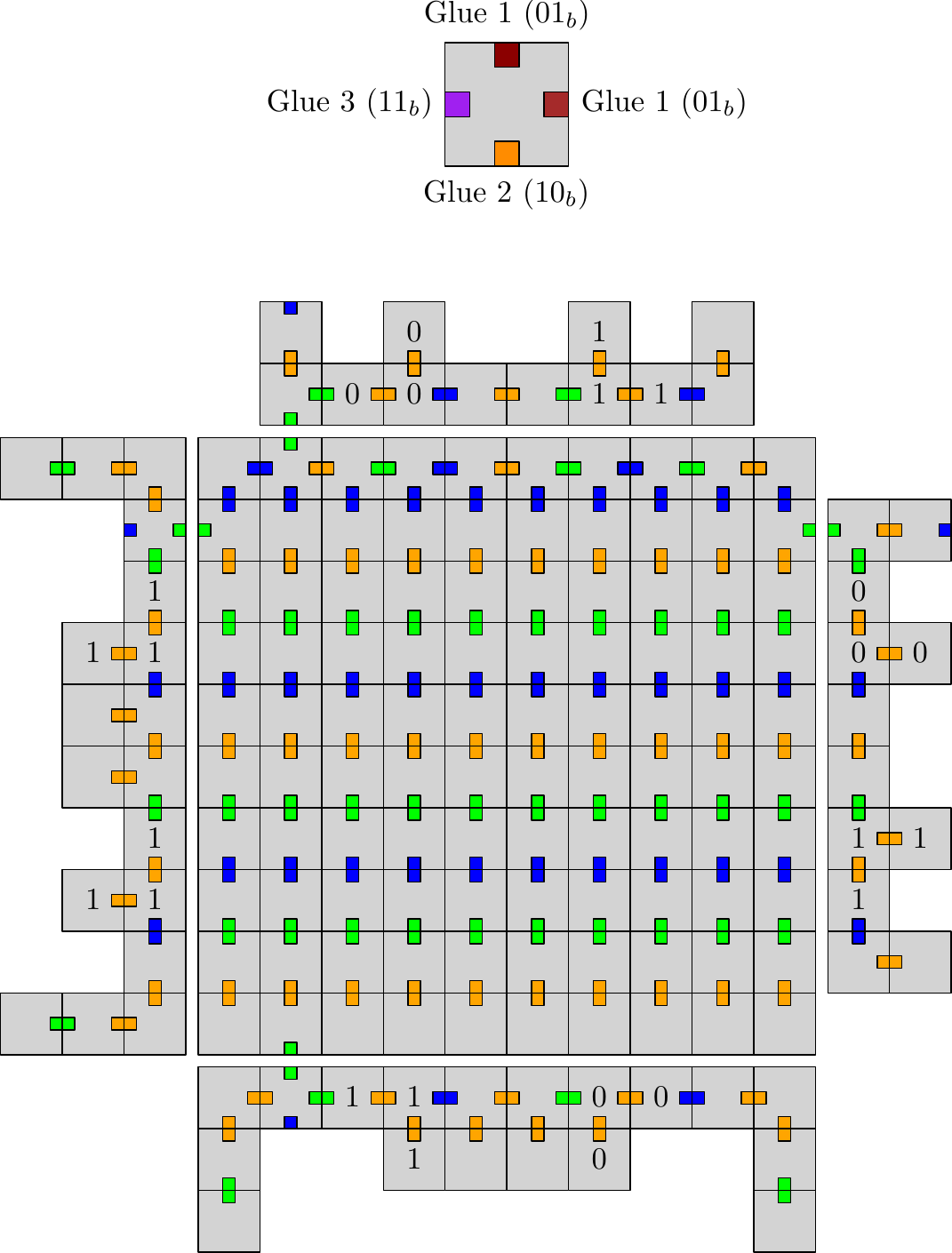}
\caption{Converting a tile in a system with 7 glues to a macrotile with $O(\log{|G|})$ scale and 3 glues.
The gray label of the tile is used as a label for all tiles in the core and macroglue assemblies, with the 1 and 0 markings for illustration of the glue bit encoding.}
\label{fig:glue-bit-encoding}
\end{figure}

For each label $l \in \Sigma(T)$, repeat the process of producing the macroglue assemblies once using a tile set exclusively labeled $l$.
Also produce a square $\Theta(\log{|G|}) \times \Theta(\log{|G|})$ \emph{core assembly}, with a single copy of the inner glue on the counterclockwise end of each face.
Use the macroglue and core assemblies to produce a set of \emph{macrotiles}, one for each $t \in T$, consisting of a core assembly whose tiles have the label of $t$, and four glue assemblies encode the four glues of $t$ and whose tiles have the label of $t$.
Extend the mix graph $M'$ of $\mathcal{S}'$ by carrying out the mixings of $M$ but starting with the equivalent macrotiles.
Define the simulation function $f$ to map each macrotile to the label found on the macrotile, and the function $g$ to take the portion of $M'$ and $g$ to be the portion of the mix graph carrying out the mixings of $\mathcal{S}$.
 
The work done to produce the glue assemblies is $O(\Sigma(T)|G|)$, to produce the core assemblies is $O(\Sigma(T)\log\log{|G|})$, and to produce the macrotiles is $O(|T|)$.
Carrying out the mixings of $\mathcal{S}$ requires $O(|\mathcal{S}|)$ work.
Since each macrotile is used in at least one mixing simulating a mixing in $\mathcal{S}$, $|T| \leq |\mathcal{S}|$.
Additionally, $|G| \leq 4|T|$.
So the total system size is $O(\Sigma(T)|G| + \Sigma(T)\log\log{|G|} + |T| + |\mathcal{S}|) = O(\Sigma(T)|T| + |\mathcal{S}|)$. 
\end{proof}
}


Armed with these tools, we are ready to convert PCFGs into SSASs.
Recall that in Section~\ref{sec:pcfgs-slightly-better-sas-ssas} we showed that in the worst case, converting a PCFG into a SSAS (or SAS) \emph{must} incur an $\Omega(\log{n}/\log\log{n})$-factor increase in system size.
Here we achieve a $O(\log{n})$-factor increase.

\begin{theorem}
\label{thm:sas-over-pcfg-upper-bound}
For any polyomino $P$ with $|P| = n$ derived by a PCFG $G$, there exists a SSAS $\mathcal{S}$ with $|\mathcal{S}| = O(|G|\log{n})$ producing an assembly with label polyomino $P'$, where $P'$ is a $(O(\log{n}), O(n))$-fuzzy replica of $P$. 
\end{theorem}

\begin{proof}
We combine the macrotile construction of Lemma~\ref{lem:constant-glues}, the generalized counters of Lemma~\ref{lem:efficient-incrementors-sas}, and a macrotile assembly invariant that together enable efficient simulation of each production rule in a PCFG by a set of $O(\log{n})$ mixing steps.

\textsc{Macrotiles}. The macrotiles used are extended versions of the macrotiles in Lemma~\ref{lem:constant-glues} with two modifications: a secondary, \emph{resevoir macroglue} assembly on each side of the tile in addition to a primary \emph{bonding macroglue}, and a thin \emph{cage} of dimensions $\Theta(n) \times \Theta(\log{n})$ surrounding each resevoir macroglue (see Figure~\ref{fig:conversion-macrotile}).

\begin{figure}[ht]
\centering
\includegraphics[scale=1.0]{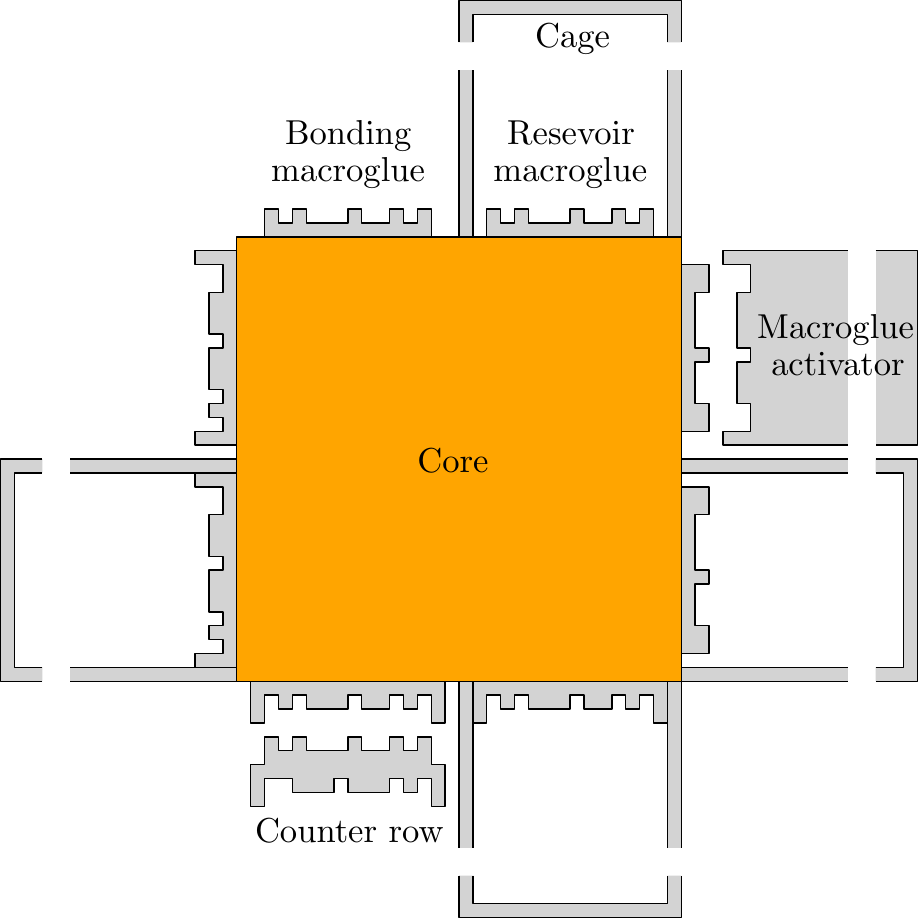}
\caption{A macrotile used in converting a PCFG to a SAS, and examples of value maintenance and offset preparation.}
\label{fig:conversion-macrotile}
\end{figure}
 
Mixing a macrotile with a set of bins containing counter row assemblies constructed by Lemma~\ref{lem:efficient-incrementors-sas} causes completed (and incomplete) counter rows to attach to the macrotile's macroglues.
Because each macroglue's geometry matches the geometry of exactly one counter row, a partially completed counter row that attaches can only be completed with bit assemblies that match the macroglue's value.
As a result, mixing the bin sets of Lemma~\ref{lem:efficient-incrementors-sas} with an assembly consisting of macrotiles produces the same set of products as mixing a completed set of binary counter rows with the assembly.

An attached counter row effectively causes the macroglue's value to change, as it presents geometry encoding a new value and covers the macroglue's previous value. 
The cage is constructed to have height sufficient to accomodate up to $n$ counter rows attached to the reservoir macroglue, but no more.

Because of the cage, no two macrotiles can attach by their bonding macroglues unless the macroglue has more than $n$ counter rows attached.
Alternatively, one can produce a thickened counter row with thickness sufficient to extend beyond the cage.
We call such an assembly a \emph{macroglue activator}, as it ``activates'' a bonding macroglue to being able to attach to another promoted macroglue on another macrotile.
Notice that a macroglue activator will never attach to a bonding macroglue's resevoir twin, as the cage is too small to contain the activator.

\textsc{An invariant}. Counter rows and activators allow precise control of two properties of a macrotile: the identities of the macroglues on each side, and whether these glues are activated.
In a large assembly containing many macroglues, the ability to change and activate glues allows precise encoding of how an assembly can attach to others.
In the remainder of the construction we maintain the invariant that every macrotile has the same glue identity on all four sides, and any macrotile assembly consists of macrotiles with glue identities forming a contiguous interval, e.g. $4$, $5$, $6$, $7$.
Intervals are denoted $[i, i']$, e.g. $[g_4, g_7]$.

By Lemma~\ref{lem:efficient-incrementors-sas}, a set of row counters incrementing the glue identities of \emph{all} glues on a macrotile can be produced using $O(\log{n})$ work.
Activators, by virtue of being nearly rectangular with $O(\log{n})$ cells of bit geometry can also be produced using $O(\log{n})$ work.

\textsc{Production rule simulation}. Consider a PCFG with non-terminal $N$ and production rule $N \rightarrow (R_1, (x_1, y_1)) (R_2, (x_2, y_2))$ and a SSAS with two bins containing assemblies $A_1$, $A_2$ with the label polyominoes of $A_1$ and $A_2$ being fuzzy replicas of the polyominoes derived by $R_1$ and $R_2$.
Also assume $A_1$ and $A_2$ are assembled from the macrotiles just described, including the invariant that the identities of the glues on $A_1$ and $A_2$ are identical on all sides of a macrotile and contiguous across the assembly, i.e. the identities of the glues are $[i_1, j_1]$ and $[i_2, j_2]$ on assemblies $A_1$ and $A_2$, respectively. 

Select two cells $c_{R_1}$, $c_{R_2}$, in the polyominoes derived by $R_1$ and $R_2$ adjacent in polyomino derived by $N$.
Define the glue identities of the two macrotiles forming the supercells mapped to $c_{R_1}$ and $c_{R_2}$ to be $g_1$ and $g_2$.
Then the glue sets on $A_1$ and $A_2$ can be decomposed into three subsets $[i_1, g_1-1]$, $[g_1]$, $[g_1+1, j_1]$ and $[i_2, g_2-1]$, $[g_2]$, $[g_2+1, j_2]$, respectively.
We change these glue values in three steps:

\begin{enumerate}
\item Construct two sets of row counters that increment $i_1$ through $g_1$ by $j_1 - i_1 + 1$ and $i_2$ through $g_2$ by $g_2 - i_2 + 1$, and mix them in separate bins with $A_1$ and $A_2$ to produce two new assemblies $A_1'$ and $A_2'$.
Assemblies $A_1'$ and $A_2'$ have glues $[g_1 + 1, g_1 + j_1 - i_1 + 1]$ and $[g_2, g_2 + j_2 - i_2]$, respectively, and the macroglues with values $g_1$ and $g_2$ now have values $g_1' = g_1 + (g_1 - i_1) + j_1 + 1$ and $g_2' = g_2$, i.e. the glues of $A_1'$ and $A_2'$ are $[g_1' - (j_1 - i_1), g_1']$ and $[g_2', g_2' + j_2 - i_2]$.
\item Construct a set of row counters that increment the values of all glues on $A_2'$ by $g_2' - g_1' + 1$ if this value is positive, and mix the counters with $A_2'$ to produce $A_2''$.
Then the macroglue with value $g_2'$ now has value $g_2'' = g_1' + 1$ and the glue values of $A_1'$ and $A_2''$ are $[g_1' - (j_1 - i_1), g_1']$ and $[g_2'', g_2'' + j_2 - i_2]$.
\item Construct a pair of macroglue activators with values $g_1'$ and $g_2''$ that attach to the pair of macroglue sides matching the two adjacent sides of cells $c_{R_1}$ and $c_{R_2}$. 
Mix each activator with the corresponding assembly $A_1'$ or $A_2''$.
\end{enumerate}

Mixing $A_1'$ and $A_2''$ with the pair of activated macroglues causes them to bond in exactly one way to form a superassembly $A_3$ whose label polyomino is a fuzzy replica of the polyomino derived by $N$.
Moreover, the glue values of the macrotiles in $A_3$ are $[g_1' - (j_1 - i_1), g_2'' + j_2 - i_2]$, maintaining the invariant.
Because each macrotile has a resevoir macroglue on each side, any bonding macroglue with an activator already attached has a resevoir macroglue that accepts the matching row counter, so each mixing has a single product and specifically no row counter products.

\textsc{System scale} The PCFG $P$ contains at most $n$ production rules.
Also, each step shifts glue identities by at most $n$ (the number of distinct glues on the macrotile), so the largest glue identity on the final macrotile assembly is $n^2$.
So we produce macrotiles with core assemblies of size $O(\log{n}) \times O(\log{n})$ and cages of size $O(n)$.
Assembling the core assemblies, cages, and initial macroglue assemblies of the macrotiles takes $O(|P|\log{n} + \log{n} + \log{n}) = O(|P|\log{n})$ work, dominated by the core assembly production.
Simulating each production rule of the grammar takes $O(\log{n})$ work spread across a constant number of $O(\log{n})$-sized sequences of mixings to produce sets of row counters and macroglue activators.

\end{proof} 

Applying Lemma~\ref{lem:constant-glues} to the construction (creating macrotiles of macrotiles) gives a constant-glue version of Theorem~\ref{thm:sas-over-pcfg-upper-bound}:

\both{
\begin{theorem}
\label{thm:sas-over-pcfg-upper-bound-constant-glues}
For any polyomino $P$ with $|P| = n$ derived by a PCFG $G$, there exists a SSAS $\mathcal{S'}$ using $O(1)$ glues with $|\mathcal{S'}| = O(|G|\log{n})$ producing an assembly with label polyomino $P'$, where $P'$ is a $(O(\log{n}\log\log{n}), O(n\log\log{n}))$-fuzzy replica of $P$. 
\end{theorem}
}

\later{
\begin{proof}
The construction of Theorem~\ref{thm:sas-over-pcfg-upper-bound} uses $O(\log{n})$ glues, namely for the counter row subconstruction of Lemma~\ref{lem:efficient-incrementors-sas}.
With the exception of the core assemblies, all tiles of $S$ have a common fuzz (gray) label, so creating macrotile versions of these tiles and carrying out all mixings involving these macrotiles and \emph{completed} core assemblies is possible with $O(1 \cdot |T| + |\mathcal{S}|) = O(|\mathcal{S}|)$ mixings and scale $O(\log\log{n})$.
Scaled core assemblies of size $\Theta(n\log\log{n}) \times \Theta(n\log\log{n})$ can be constructed using constant glues and $O(\log(n\log\log{n})) = O(\log{n})$ mixings, the same number of mixings as the unscaled $\Theta(n) \times \Theta(n)$ core assemblies of Theorem~\ref{thm:sas-over-pcfg-upper-bound}.
So in total, this modified construction has system size $O(|S|) = O(|G|\log{n})$ and scale $O(\log\log{n})$.
Thus it produces an assembly with label polyomino that is a $(O(\log{n}\log\log{n}), O(n\log\log{n}))$-fuzzy replica of $P$.
\end{proof}
}

The results in this section and Section~\ref{sec:pcfgs-slightly-better-sas-ssas} achieve a ``one-sided'' correspondence between the smallest PCFG and SSAS encoding a polyomino, i.e. the smallest PCFG is approximately an \emph{upper bound} for the smallest SSAS (or SAS).
Since the separation upper bound proof (Theorem~\ref{thm:sas-over-pcfg-upper-bound}) is constructive, the bound also yields an algorithm for converting a a PCFG into a SSAS.

\section{PCFG over SAS and SSAS Separation Lower Bound}
\label{sec:sass-much-better-than-pcfgs}

Here we develop a sequence of PCFGs over SAS and SSAS separation results, all within a polylogarithmic factor of optimal.
The results also hold for polynomially scaled versions of the polyominoes, which is used to prove Theorem~\ref{thm:constant-glues} at the end of the section.
This scale invariance also surpasses the scaling of the fuzzy replicas in Theorems~\ref{thm:sas-over-pcfg-upper-bound} and~\ref{thm:sas-over-pcfg-upper-bound-constant-glues}, implying that this relaxation of the problem statement in these theorems was not unfair.

\begin{figure}[ht]
\centering
\includegraphics[scale=1.0]{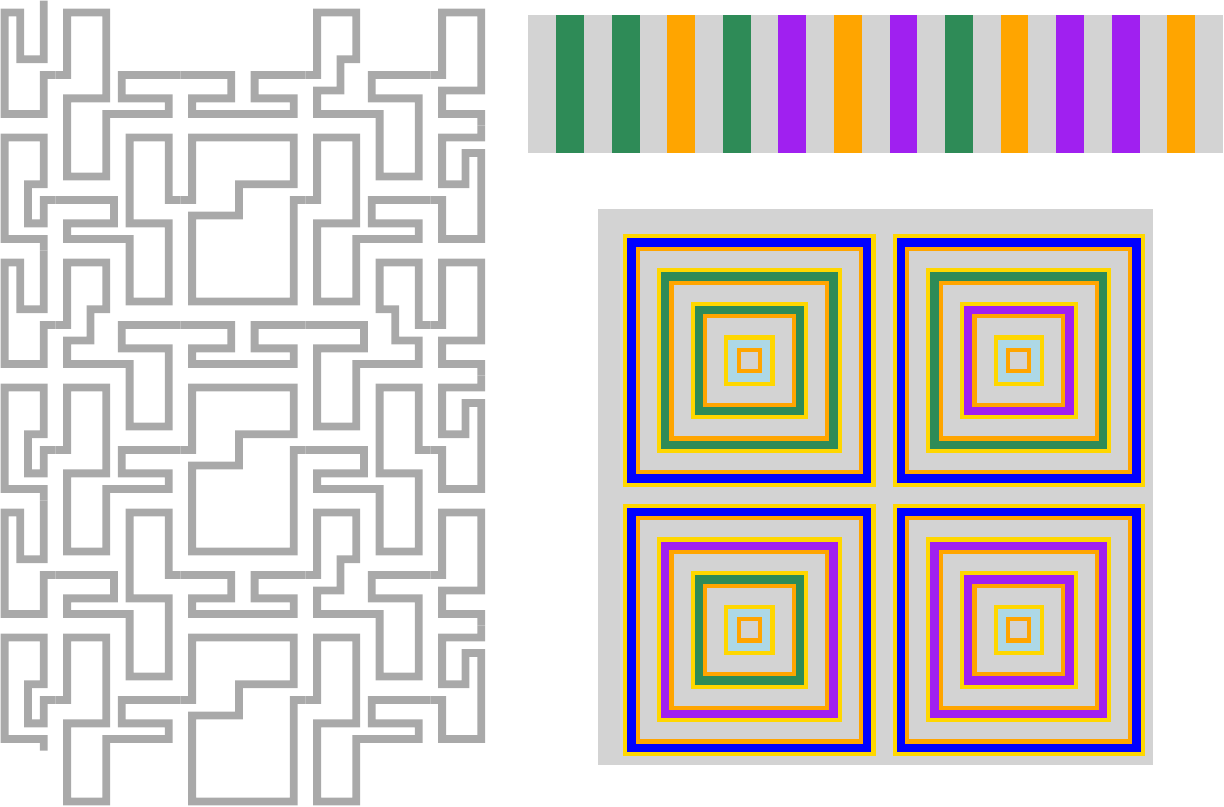}
\caption{Two-bit examples of the weak (left), end-to-end (upper right), and block (lower right) binary counters used to achieve separation of PCFGs over SASs and SSASs in Section~\ref{sec:sass-much-better-than-pcfgs}.}
\label{fig:taste}
\end{figure}

\subsection{General shapes}
\label{sec:sas-much-better-pcfgs-gen}

\later{
In this section we describe an efficient system for assembling a set of shapes we call \emph{weak counters}.
An example of a rows in the original counter and macrotile weak counter are shown in Figure~\ref{fig:weak-counter-ex}.
These shapes are macrotile versions of the doubly-exponential counters found in~\cite{Demaine-2008} with three modifications:

\begin{enumerate}
\item Each row is a single path of tiles, and any path through an entire row uniquely identifies the row.
\item Adjacent rows do not have adjacent pairs of tiles, i.e. they do not touch.
\item Consecutive rows attach at alternating (east, west, east, etc.) ends.
\end{enumerate}

\begin{figure}[ht]
\centering
\includegraphics[width=\columnwidth]{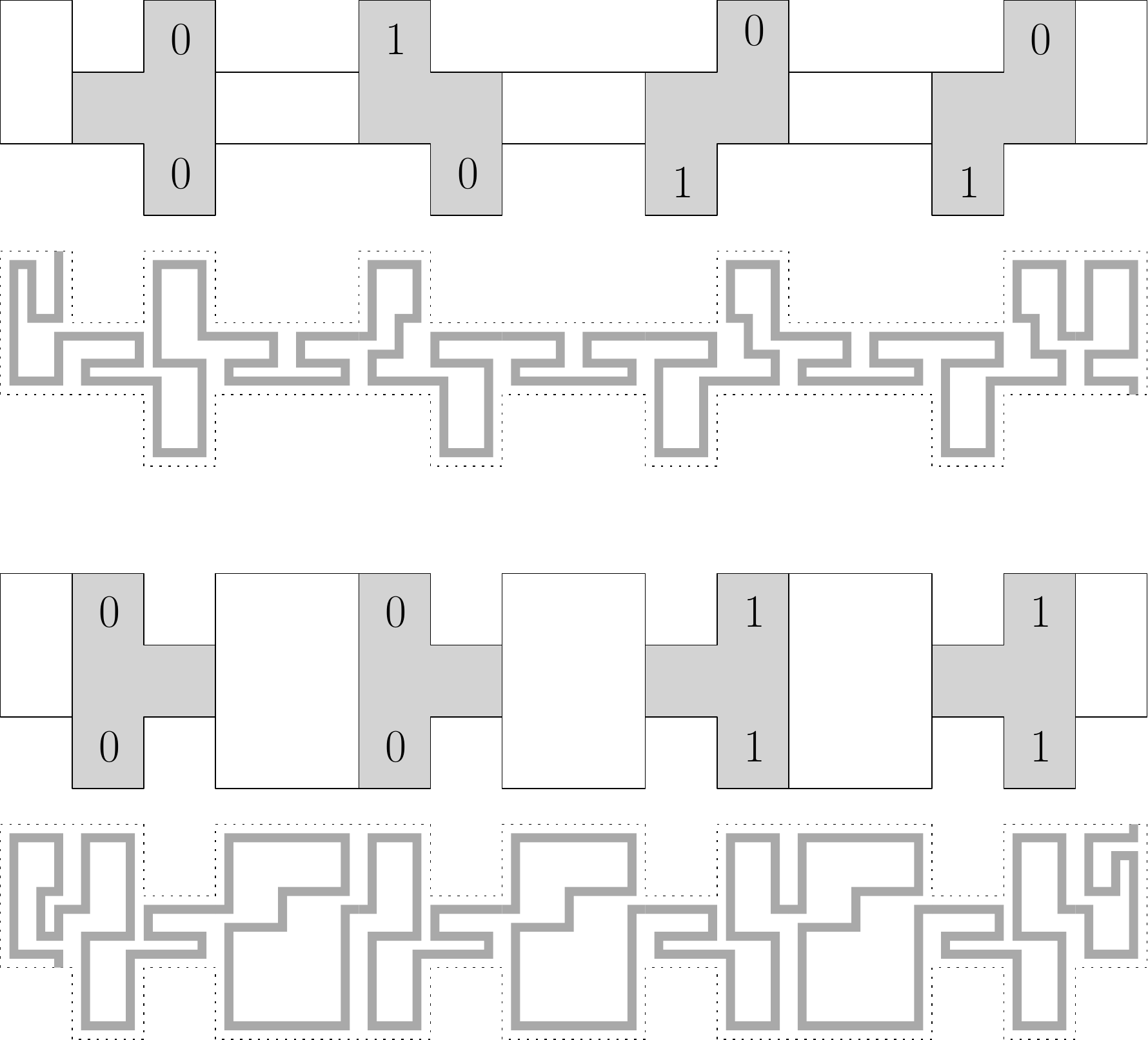}
\caption{Zoomed views of increment (top) and copy (bottom) counter rows described in~\cite{Demaine-2008} and the equivalent rows of a weak counter.}
\label{fig:weak-counter-ex}
\end{figure}

Figure~\ref{fig:weak-counter-mate} shows three consecutive counter rows attached in the final assembly.
Each row of the doubly-exponential counter consists of small, constant-sized assemblies corresponding to 0 or 1 values, along with a 0 or 1 carry bit.
We implement each assembly as a unique path of tiles and assemble the counter as in~\cite{Demaine-2008}, but using these path-based assemblies in place of the original assemblies.
We also modify the glue attachments to alternate on east and west ends of each row.
Because the rows alternate between incrementing a bit string, and simply encoding it, alternating the attachment end is trivial.
Finally, note that adjacent rows only touch at their attachment, but the geometry encoded into the row's path prevents non-consecutive rows from attaching.

\begin{figure}[ht]
\centering
\includegraphics[width=1.0\columnwidth]{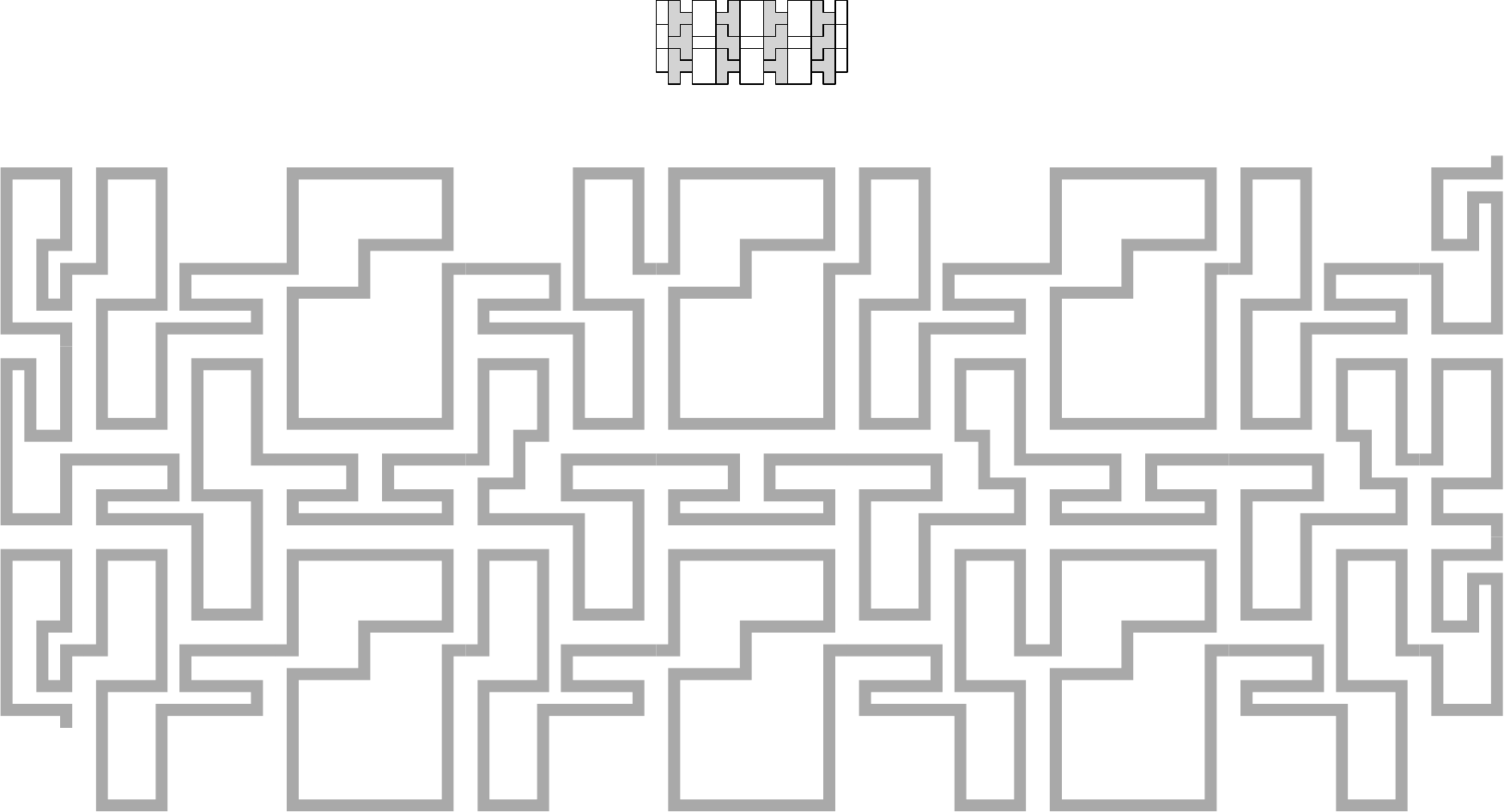}
\caption{Adjacent attached rows of the counter described in~\cite{Demaine-2008} (top) and the equivalent rows in the weak counter (bottom).}
\label{fig:weak-counter-mate}
\end{figure}
}

\both{
\begin{lemma}
\label{lem:sas-weak-counter-const}
There exists a $\tau = 1$ SAS of size $O(b)$ that produces a $2^b$-bit weak counter.
\end{lemma}
}

\later{
\begin{proof}
The counter is an $O(1)$-scaled version of the counter of Demaine et al~\cite{Demaine-2008}.
They show that such an assembly is producible by a system of size $O(b)$.
\end{proof}
}

\both{
\begin{lemma}
\label{lem:pcfg-weak-counter-lower-bound}
For any PCFG $G$ deriving a $2^b$-bit weak counter, $|G| = \Omega(2^{2^b})$.
\end{lemma}
}

\later{
\begin{proof}

Define a \emph{minimal row spanner} of row $\row_i$ to be a non-terminal symbol $N$ of $G$ with production rule $N \rightarrow (B, (x_1, y_1)) (C, (x_2, y_2))$ such that the polyomino derived by $N$ contains a path between a pair of easternmost and westernmost tiles of the row and the polyominoes derived by $B$ and $C$ do not.
We claim that each row (trivially) has at least one minimal row spanner and each non-terminal of $G$ is a minimal row spanner of at most one unique row.

First, suppose by contradiction that a non-terminal $N$ is a minimal row spanner for two distinct rows.
Because $N$ is connected and two non-adjacent rows are only connected to each other via an intermediate row, $N$ must be a minimal row spanner for two adjacent rows $\row_i$ and $\row_{i+1}$.
Then the polyominoes of $B$ and $C$ each contain tiles in both $\row_i$ and $\row_{i+1}$, as otherwise either $C$ or $B$ is a minimal row spanner for $\row_i$ or $\row_{i+1}$.

Without loss of generality, assume $B$ contains a tile at the end of $\row_i$ not adjacent to $\row_{i+1}$.
But $B$ also contains a tile in $\row_{i+1}$ and (by definition) is connected.
So $B$ contains a path between the east and west ends of row $\row_i$, and thus $N$ is a not a minimal row spanner for $r_i$.
So $N$ is a minimal row spanner for at most one row.

Next, note that the necessarily-serpentine path between a pair of easternmost and westernmost tiles of a row in a minimal row spanner uniquely encodes the row it spans.
So the row spanned by a minimal row spanner is unique.

Because each non-terminal of $G$ is a minimal row spanner for at most one unique row, $G$ must have at least $2^{2^b}$ non-terminal symbols and total size $\Omega(2^{2^b})$.
\end{proof}
}

\begin{theorem}
\label{thm:pcfgs-over-sas-single-label-gen}
The separation of PCFGs over $\tau = 1$ SASs for single-label polyomines is $\Omega(n/(\log\log{n})^2)$.
\end{theorem}

\begin{proof}
By the previous two lemmas, there exists a SAS of size $O(b)$ producing a $b$-bit weak counter, and any PCFG deriving this shape has size $\Omega(2^{2^b})$.
The assembly itself has size $n = \Theta(2^{2^b}b)$, as it consists of $2^{2^b}$ rows, each with $b$ subassemblies of constant size.
So the separation is $\Omega((n/b)/b) = \Omega(n/(\log\log{n})^2)$.
\end{proof}

\later{
In~\cite{Demaine-2008}, the $O(\log\log{n})$-sized SAS constructing a $\log{n}$-bit binary counter repeatedly doubles the length of each row (i.e. number of bits in the counter) using $O(1)$ mixings per doubling.
Achieving such a technique in a SSAS seems impossible, but a simpler construction producing a $b$-bit counter with $O(b)$ work can be done by using a unique set of $O(1)$ glues for each bit of the counter.
In this case, mixing these reusable elements along with a previously-constructed pair of first and last counter rows creates a single mixing assembling the entire counter at once.
Modifying the proof of Theorem~\ref{thm:pcfgs-over-sas-single-label-gen} to use this construction gives a similar separation for SSASs:
}

\both{
\begin{corollary}
The separation of PCFGs over $\tau = 1$ SSASs for single-label polyominoes is $\Omega(n/\log^2{n})$.
\end{corollary}
}

\subsection{Rectangles}
\label{sec:sas-much-better-pcfgs-rectangle}

For the weak counter construction, the lower bound in Lemma~\ref{lem:pcfg-weak-counter-lower-bound} depended on the poor connectivity of the weak counter polyomino.
This dependancy suggests that such strong separation ratios may only be achievable for special classes of ``weakly connected'' or ``serpentine'' shapes.
Restricting the set of shapes to rectangles or squares while keeping an alphabet size of 1 gives separation of at most $O(\log{n})$, as any rectangle of area $n$ can be derived by a PCFG of size $O(\log{n})$.

But what about rectangles with a constant-sized alphabet?
In this section we achieve surprisingly strong separation of PCFGs over SASs and SSASs for rectangular constant-label polyominoes, nearly matching the separation achieved for single-label general polyominoes.

\later{

\paragraph{The construction}

The polyominoes constructed resemble binary counters whose rows have been arranged in sequence horizontally, and we call them \emph{$b$-bit end-to-end counters}.
Each row of the counter is assembled from tall, thin macrotiles (called \emph{bars}), each containing a \emph{color strip} of orange, purple, or green.
The color strip is coated on its east and west faces with gray geometry tiles that encode the bar's location within the counter.

\begin{figure}[ht]
\centering
\includegraphics[scale=1.0]{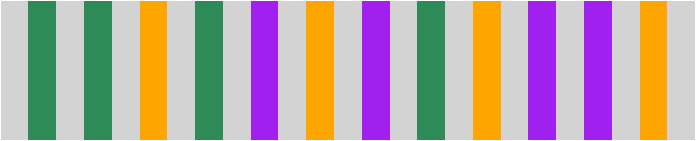}
\caption{The rectangular polyomino used to show separation of PCFGs over SASs when constrained to constant-label rectangular polyominoes.
The green and purple color strips denote 0 and 1 bits in the counter.}
\label{fig:rect-pattern}
\end{figure}

Each row of the counter has a sequence of green and purple \emph{display bars} encoding a binary representation of the row's value and flanked by orange \emph{reset bars} (see Figure~\ref{fig:rect-structure}).
An example for $b = 2$ bits can be seen in Figure~\ref{fig:rect-pattern}.

\begin{figure}[ht]
\centering
\includegraphics[scale=1.0]{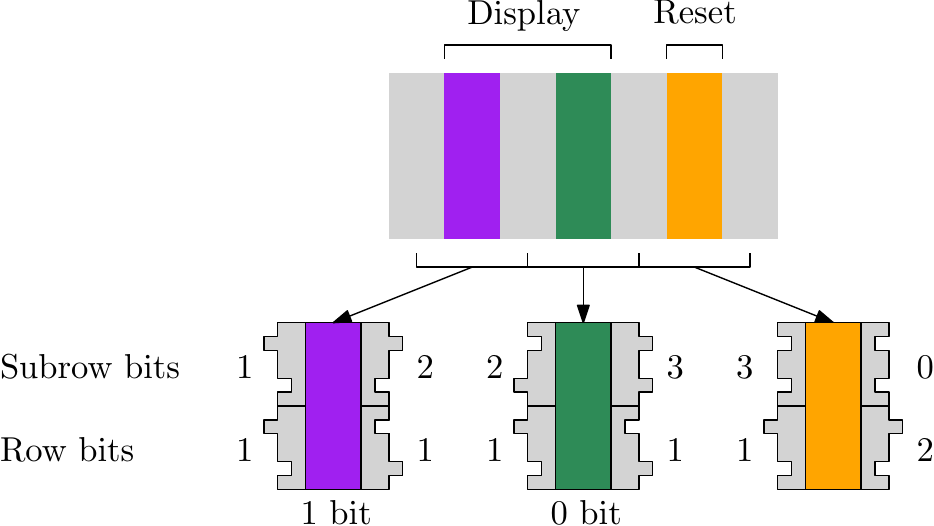}
\caption{The implementation of the vertical bars in row 2 ($01_b$) of an end-to-end counter.}
\label{fig:rect-structure}
\end{figure}

Each bar has dimensions $O(1) \times 3(\log_2{b} + b + 2)$, sufficient for encoding two pieces of information specifying the location of the bar within the assembly.
The \emph{row bits} specify which row the bar lies in (e.g. the $7$th row).
The \emph{subrow bits} specify where within the row the bar lies (e.g. the $4$th bit).
The subrow value starts at $0$ on the east side of a reset bar, and increments through the display bars until reaching $b+1$ on the west end of the next reset bar.
Bars of all three types with row bits ranging from $0$ to $2^b-1$ are produced.

\paragraph{Efficient assembly}

The counter is constructed using a SAS of size $O(b)$ in two phases.
First, sequences of $O(b)$ mixings are used to construct five families of bars:
1. reset bars,
2. 0-bit display bars resulting from a carry,
3. 0-bit display bars without a carry,
4. 1-bit display bars resulting from a carry,
5. 1-bit display bars without a carry,
The mixings product five bins, each containing \emph{all} of the bars in the family.
These five bins are then combined into a final bin where the bars attach to form the $\Theta(2^b) \times \Theta(b)$ rectangular assembly.
The five families are seen in Figure~\ref{fig:rect-schema}.

\begin{figure}[ht]
\centering
\includegraphics[scale=1.0]{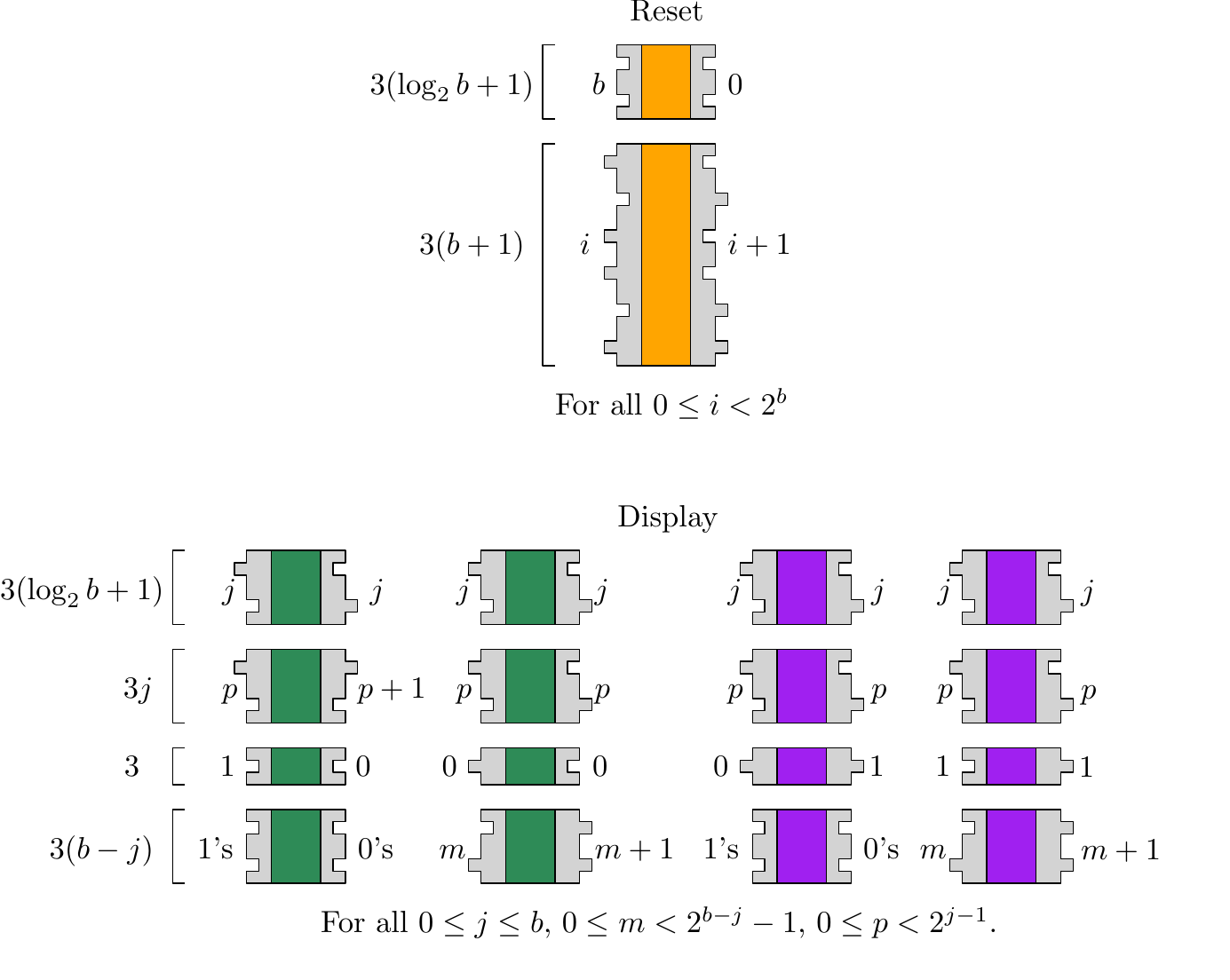}
\caption{The decomposition of bars used assemble a $b$-bit end-to-end counter.}
\label{fig:rect-schema}
\end{figure}

Efficient $O(b)$ assembly is achieved by careful use of the known approach of non-deterministic assembly of single-bit assemblies as done in~\cite{Demaine-2008}.
Assemblies encoding possible input bit and carry bit value combinations for each row bit and subrow bit are constructed and mixed together, and the resulting products are every valid set of input and output bit strings, i.e. every row of a binary counter assembly.

As a warmup, consider the assembly of all reset bars.
For these bars, the west subrow bits encode $b$ and the east subrow bits encode $0$.
The row bits encode a value $i$ on the west side, and $i+1$ on the east side, for all $i$ between $0$ and $2^b-1$.
Constructing all such bars using $O(b)$ work is straightforward.
For each of the $\log_2{b} + 1$ subrow bits, create an assembly where the west and east bits are 1 and 0 respectively, \emph{except} for the most significant bit (bit $\log_2{b}+1$), where the west and east bits are both 0.

For the row bits we use the same technique as in~\cite{Demaine-2008} and extended in Lemma~\ref{lem:efficient-incrementors-value-sas}: create a constant-sized set of assemblies for each bit that encode input and output value and carry bits.
For bits $1$ through $b-1$ (zero-indexed) create four assemblies corresponding to the four combinations of value and carry bits, for bit 0 create two assemblies corresponding to value bits (the carry bit is always 1), for bit $b$ create three assemblies corresponding to all combinations except both value and carry bits valued 1, and for bit $b+1$ create a single assembly with both bits valued 0.
Give each bit assembly a unique south and north glue encoding its location within the bar and carry bit value, and give all bit assemblies a common orange color strip.
Mixing these assemblies produces all reset bars, with subrow west and east values of $b$ and $0$, and row values $i$ and $i+1$ for all $i$ from $0$ to $2^b-1$.

In contrast to producing reset bars, producing display bars is more difficult.
The challenge is achieving the correct color strip relative to the subrow and row values.
Recall that the row value $i$ locates the bar's row and the subrow value $j$ locates the bar within this row.
So the correct color strip for a bar is green if the $j$th bit of $i$ is 0, and purple if the $j$th bit of $i$ is 1.

We produce four families of display bars, two for each value of the $j$th bit of of $i$.
Each subfamily is produced by mixing a subrow assembly encoding $j$ on both east and west ends with three component assemblies of the row value: the \emph{least significant bits (LSB) assembly} encoding bits $1$ through $j-1$ of $i$, the \emph{most significant bits (MSB) assembly} encoding bits $j+1$ through $b$ of $i$, and the constant-sized $j$th bit assembly.
This decomposition is seen in the bottom half of Figure~\ref{fig:rect-schema}.

The four families correspond to the four input and carry bit values of the $j$th bit.
These values determine what collections of subassemblies should appear in the other two components of the row value.
For instance, if the input and carry bit values are both 1, then the LSB assembly must have all 1's on its west side (to set the $j$th carry bit to 1) and all 0's on its east side.
Similarly, the MSB assembly must have some value $p$ encoded on its west side and the value $p+1$ encoded on its east side, since the $j$th bit and and $j$th carry bit were both 1, so the $(j+1)$st carry bit is also 1.

Notice that each of the four families has $b$ subfamilies, one for each value of $j$.
Producing all subfamilies of each family is possible in $O(b)$ work by first recursively producing a set of $b$ bins containing successively larger sets of MSB and LSB assemblies for the family.
Then each subfamily can be produced using $O(1)$ amortized work, mixing one of $b$ sets of LSB assembly subfamilies, one of $b$ sets of MSB assemblies, and the $j$th bit assembly together.
For instance, one can produce the set of $b$ sets of MSB assemblies encoding pairs of values $p$ and $p+1$ on bits $b-1$ through $b$, $b-2$ through $b$, etc. by producing the set on bits $k$ through $b$, then adding four assemblies to this bin (those encoding possible pairs of inputs to the $(k-1)$st bit) to produce a similar set on bits $k-1$ through $b$.
}

\both{
\begin{lemma}
\label{lem:sas-rect-ub}
There exists a $\tau = 1$ SAS of size $O(b)$ that produces a $b$-bit end-to-end counter.
\end{lemma}
}

\later{
\begin{proof}
This follows from the description of the system.
The five families of bars can each be produced with $O(b)$ work and the bars can be combined together in a single mixing to produce the counter.
So the system has total size $O(b)$.
\end{proof}
}

\both{
\begin{lemma}
\label{lem:pcfg-end-to-end-counter-lower-bound}
For any PCFG $G$ deriving a $b$-bit end-to-end counter, $|G| = \Omega(2^b)$.
\end{lemma}
}

\later{
\begin{proof}

Let $G$ be a PCFG deriving a $b$-bit end-to-end counter.
Define a \emph{minimal row spanner} to be a non-terminal symbol $N$ with production rule $N \rightarrow (B, (x_1, y_1)) (C, (x_2, y_2))$ such that the polyomino derived by $N$ (denoted $p_N$) horizontally spans the color strips of all bars in row $\row_i$ including the reset bar at the end of the row, while the polyominoes derived by $B$ and $C$ (denoted $p_B$ and $p_C$) do not.
Consider the bounding box $D$ of these color strips (see Figure~\ref{fig:rect-proof}).

\begin{figure}[ht]
\centering
\includegraphics[scale=1.0]{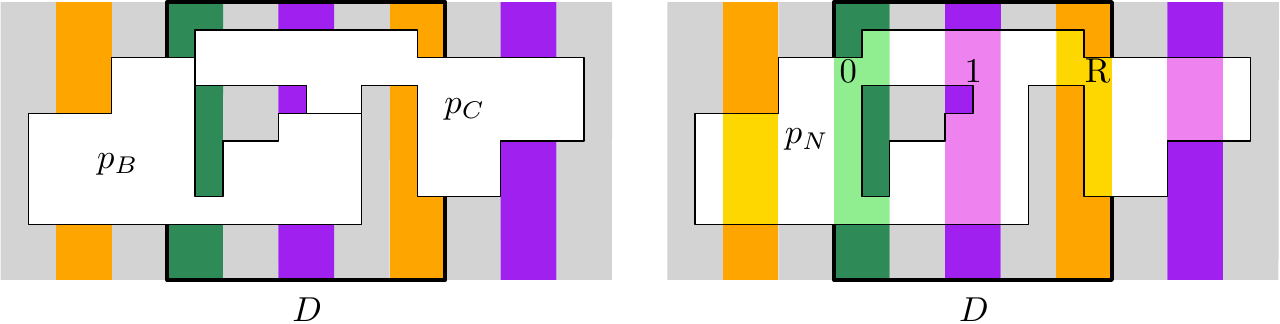}
\caption{A schematic of the proof that a non-terminal is a minimal row spanner for at most one unique row.
(Left) Since $p_B$ and $p_C$ can only touch in $D$, their union non-terminal $N$ must be a minimal row spanner for the row in $D$.
(Right) The row's color strip sequence uniquely determines the row spanned by $N$ ($01_b$).}
\label{fig:rect-proof}
\end{figure}

Without loss of generality, $p_B$ intersects the west boundary of $D$ but does not reach the east boundary, while $p_C$ intersects the east boundary but does not reach the west boundary, so any location at which $p_B$ and $p_C$ touch must lie in $D$.
Then any row spanned by $p_N$ and not spanned by $p_B$ or $p_C$ must lie in $D$, since spanning it requires cells from both $p_B$ and $p_C$.
So $p_N$ is a minimal row spanner for at most one row: row $\row_i$.

Because the sequence of green and purple display bars found in $D$ is distinct and separated by display bars in other rows by orange reset bars, each minimal row spanner spans a unique row $\row_i$.
Then since each non-terminal is a spanner for at most one unique row, $G$ must have $2^b$ non-terminal symbols and $|G| = \Omega(2^b)$.
\end{proof}
}

\both{
\begin{theorem}
The separation of PCFGs over $\tau = 1$ SASs for constant-label rectangles is $\Omega(n/\log^3{n})$.
\end{theorem}
}

\later{
\begin{proof}
By construction, a $b$-bit end-to-end counter has dimensions $\Theta(2^bb) \times \Theta(b)$.
So $n = \Theta(2^bb^2)$ and $b = \Theta(\log{n})$. 
Then by the previous two lemmas, the separation is $\Omega((n/b^2)/b) = \Omega(n/\log^3{n})$.
\end{proof}

We also note that a simple replacement of orange, green, and purple color strips with distinct horizontal sequences of black/white color substrips yields the same result but using fewer distinct labels.
}


\subsection{Squares}
\label{sec:sas-much-better-pcfgs-square}

The rectangular polyomino of the last section has exponential aspect ratio, suggesting that this shape requires a large PCFG because it approximates a patterned one-dimensional assemblies reminiscent of those in~\cite{Demaine-2012a}.
Creating a polyomino with better aspect ratio but significant separation is possible by extending the polyomino's labels vertically.
For a square this approach gives a separation of PCFGs over SASs of $\Omega(\sqrt{n}/\log{n})$, non-trivial but far worse than the rectangle.

\later{

\paragraph{The construction}

In this section we describe a polyomino that is square but contains an exponential number of distinct subpolyominoes such that each subpolyomino has a distinct ``minimal spanner'', using the language of the proof of Lemma~\ref{lem:pcfg-end-to-end-counter-lower-bound}.
These subpolyominoes use circular versions of the vertical bars of the construction in Section~\ref{sec:sas-much-better-pcfgs-rectangle} arranged concentrically rather than adjacently.
We call the polyomino a \emph{$b$-bit block counter}, and an example for $b=2$ is seen in Figure~\ref{fig:square-pattern}.

Each \emph{block} of the counter is a $\Theta(b^2) \times \Theta(b^2)$ square subpolyomino encoding a sequence of $b$ bits via a sequence of concentric rectangular \emph{rings} of increasing size.
Each ring has a \emph{color loop} encoding the value of a bit, or  the start or end of the bit sequence (the interior or exterior of the block, respectively).
The color loop actually has three subloops, with the center loop's color (green, purple, light blue, or dark blue in Fig.~\ref{fig:square-pattern}) indicating the bit value or sequence information, and two surrounding loops (light or dark orange in Fig.~\ref{fig:square-pattern}) indicating the interior and exterior sides of the loop.

\begin{figure}[ht]
\centering
\includegraphics[scale=1.0]{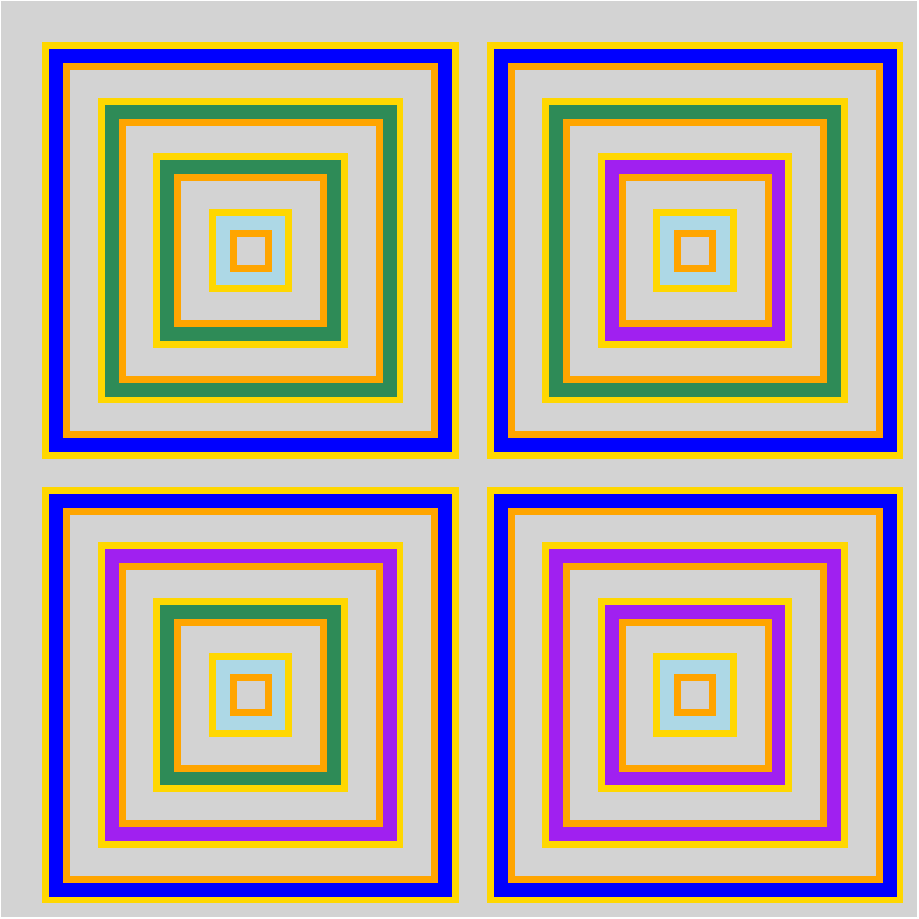}
\caption{The square polyomino used to show separation of PCFGs over SASs when constrained to constant-label square polyominoes.
The green and purple color subloops denote 0 and 1 bits in the counter, while the light and dark blue color subloops denote the start and end of the bit string.
The light and dark orange color subloops indicate the interior and exterior of the other subloops.}
\label{fig:square-pattern}
\end{figure}

\paragraph{Efficient assembly of blocks}

Though each counter block is square, they are constructed similarly to the end-to-end counter rows of Section~\ref{sec:sas-much-better-pcfgs-rectangle} by assembling the vertical \emph{bars} of each ring together into horizontal stacks of assemblies.
Horizontal \emph{slabs} are added to ``fill in'' the remaining portions of each block.

\begin{figure}[ht]
\centering
\includegraphics[scale=0.8]{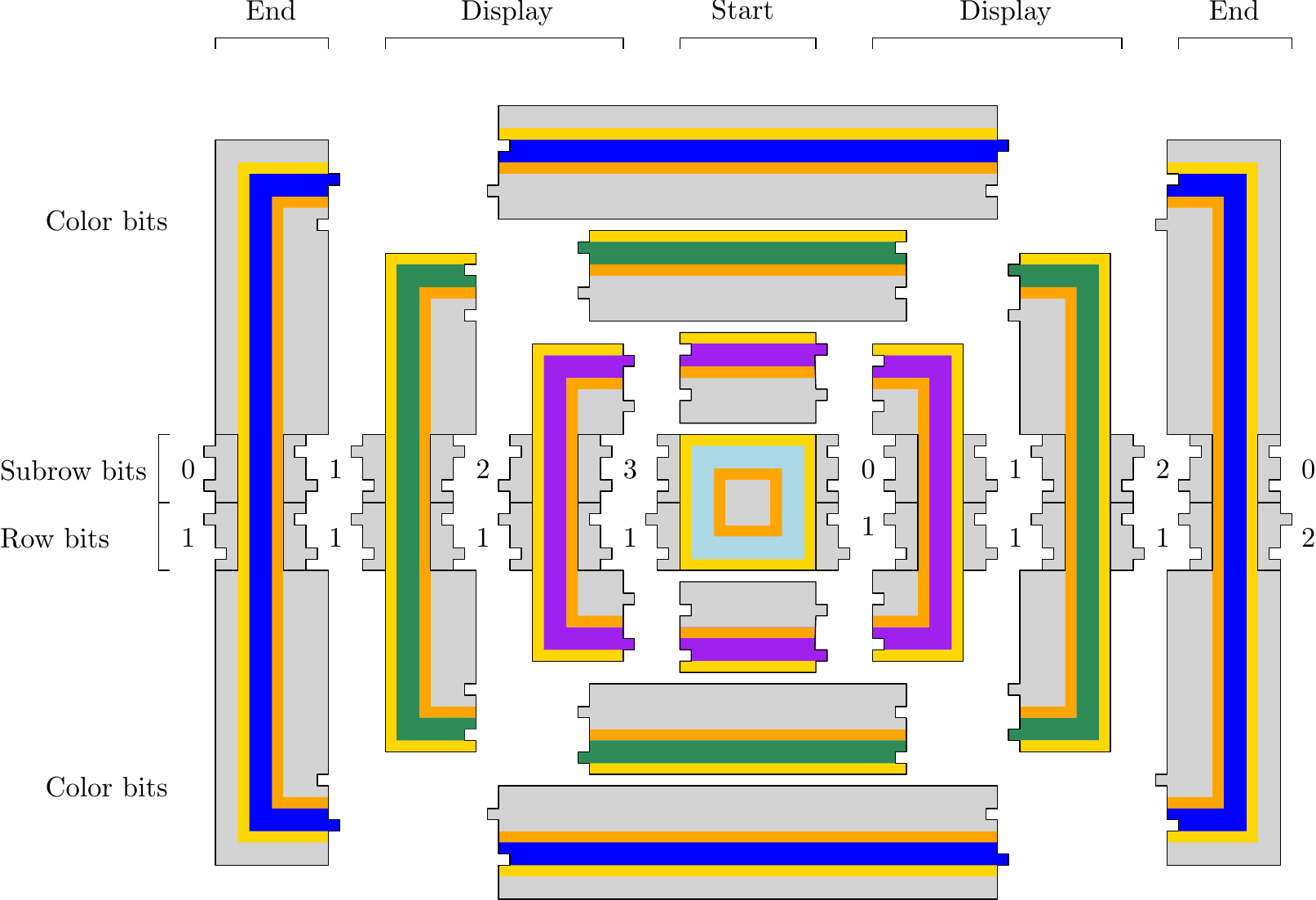}
\caption{The implementation of rings in each block of the block counter.}
\label{fig:square-structure}
\end{figure}

The bars are identical to those found in Section~\ref{sec:sas-much-better-pcfgs-rectangle} with three modifications (seen in Figure~\ref{fig:square-structure}).
First, each bar has additional height according to the value of the subrow bits (8 tiles for every increment of the bits).
Second, each bar has four additional layers of tiles on the side (east or west) facing the interior of the block, with \emph{color bits} at the north and south ends of the side encoding three values: $11_b$ (if the center color subloop is purple, a 1-bit), $00_b$ (if the center color subloop is green, a 0-bit), or $01_b$ (if the center color subloop is dark blue, the end of the bit sequence).
The additional layers are used to fill in gaps between adjacent rings left by protruding geometry, and the bit values are used to control the attachment of the horizontal slabs of each ring.

\begin{figure}[ht]
\centering
\includegraphics[scale=0.8]{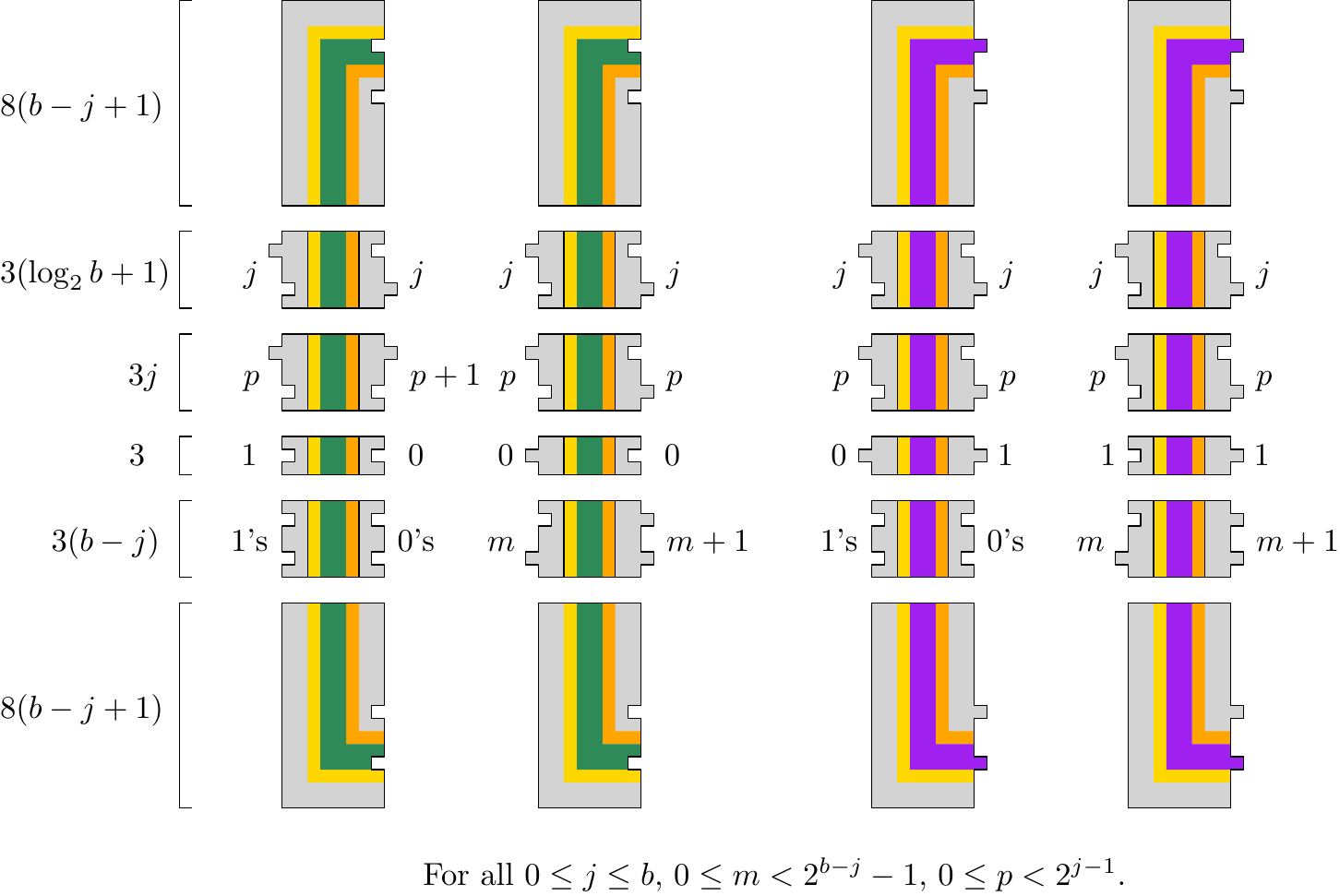}
\caption{The decomposition of vertical display bars used to assemble blocks in the $b$-bit block counter.
Only the west bars are shown, with east bars identical but color bits and color loops reflected.}
\label{fig:square-schema-display}
\end{figure}

Third, the reset bars used in Section~\ref{sec:sas-much-better-pcfgs-rectangle} are replaced with two kinds of bars: \emph{start bars} and \emph{end bars}, seein in Figure~\ref{fig:square-schema-starts-ends}.
End bars form the outermost rings of each block, and the start bars form the square cores of each block.
Both start and end bars ``reset'' the subrow counters, and the east end bars increment the row value.

\begin{figure}[ht]
\centering
\includegraphics[scale=1.0]{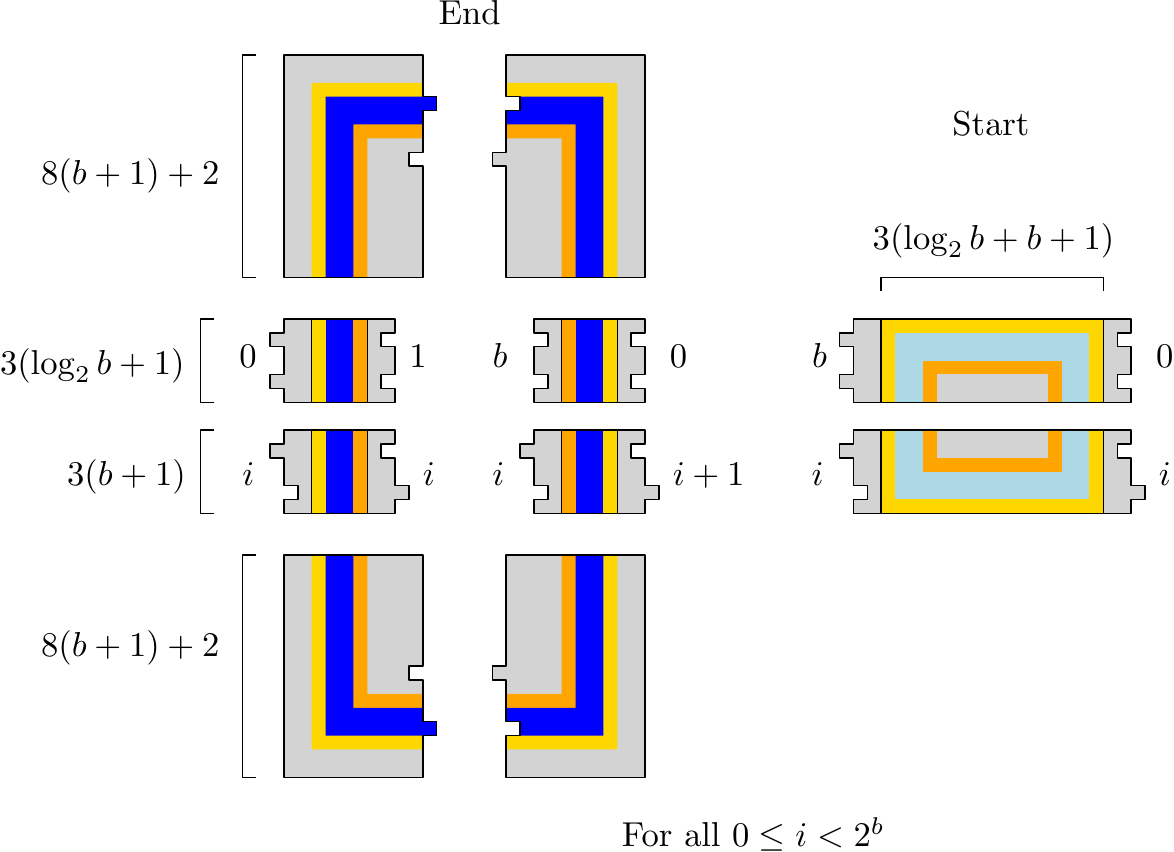}
\caption{The decomposition of vertical start and end bars used to assemble blocks in the $b$-bit block counter.}
\label{fig:square-schema-starts-ends}
\end{figure}

Recall that the vertical bars of the end-to-end counter in Section~\ref{sec:sas-much-better-pcfgs-rectangle} were constructed using $O(b)$ total work by amortizing the constructing subfamilies of MSB and LSB assemblies for each subrow value $j$.
We use the same trick here for these assemblies as well as the new assemblies on the north and south ends of each bar containing the color bits.
In total there are twelve families of vertical bar assemblies (four families of west display bars, four families of east display bars, and two families each of start and end bars), and each is assembled using $O(b)$ work.

\begin{figure}[ht]
\centering
\includegraphics[scale=1.0]{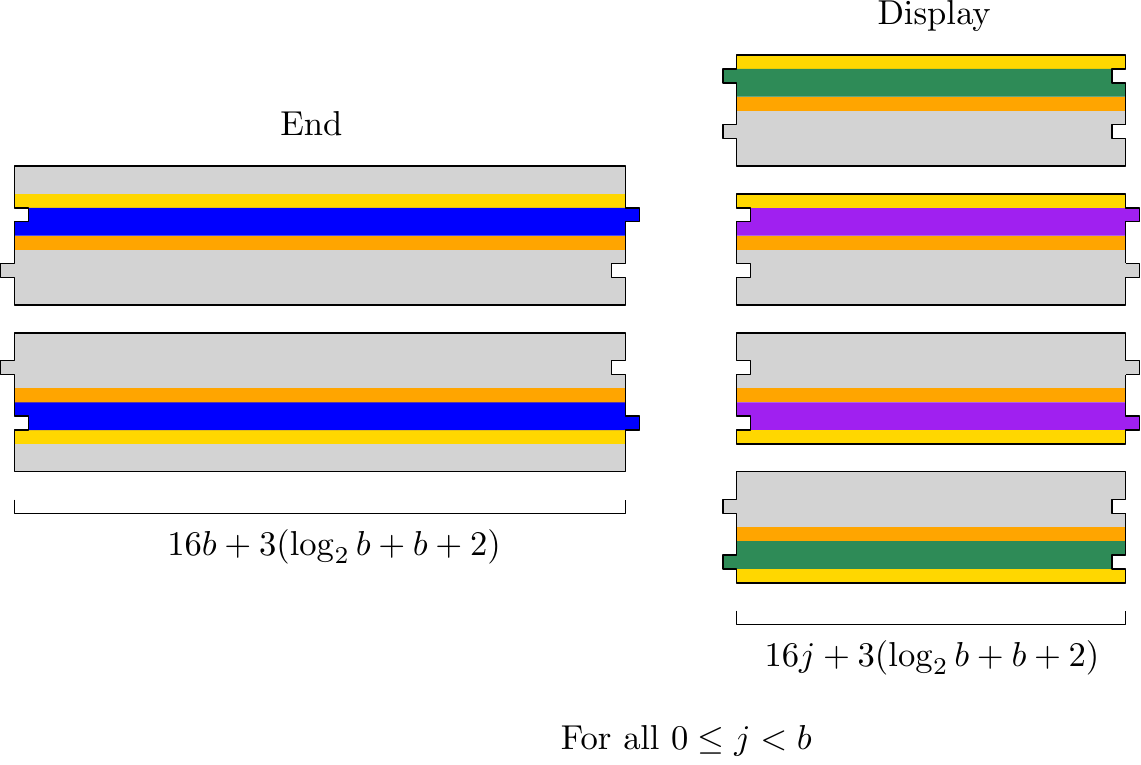}
\caption{The decomposition of horizontal slabs of each ring the $b$-bit block counter.}
\label{fig:square-schema-horizontals}
\end{figure}

Finally, the horizontal slabs of each ring are constructed as six families, each using $O(b)$ work, as seen in Figure~\ref{fig:square-schema-horizontals}.

\paragraph{Efficient assembly of the counter}

Once the families of vertical bars and horizontal slabs are assembled into blocks, we are ready to arrange them into a completed counter.
Each row of the counter has $\sqrt{2^b} = 2^{b/2}$ blocks.
So assuming $b$ is even, the $b/2$ least significant bits of the westmost block of each row are $0$'s, and of the eastmost block are $1$'s.
Before mixing the vertical bar families together, we ``cap'' the east end bar of each block at the east end of a row by constructing a set of thin assemblies (right part of Figure~\ref{fig:square-schema-caps}) and mixing them with the family of east end bars.

\begin{figure}[ht]
\centering
\includegraphics[scale=1.0]{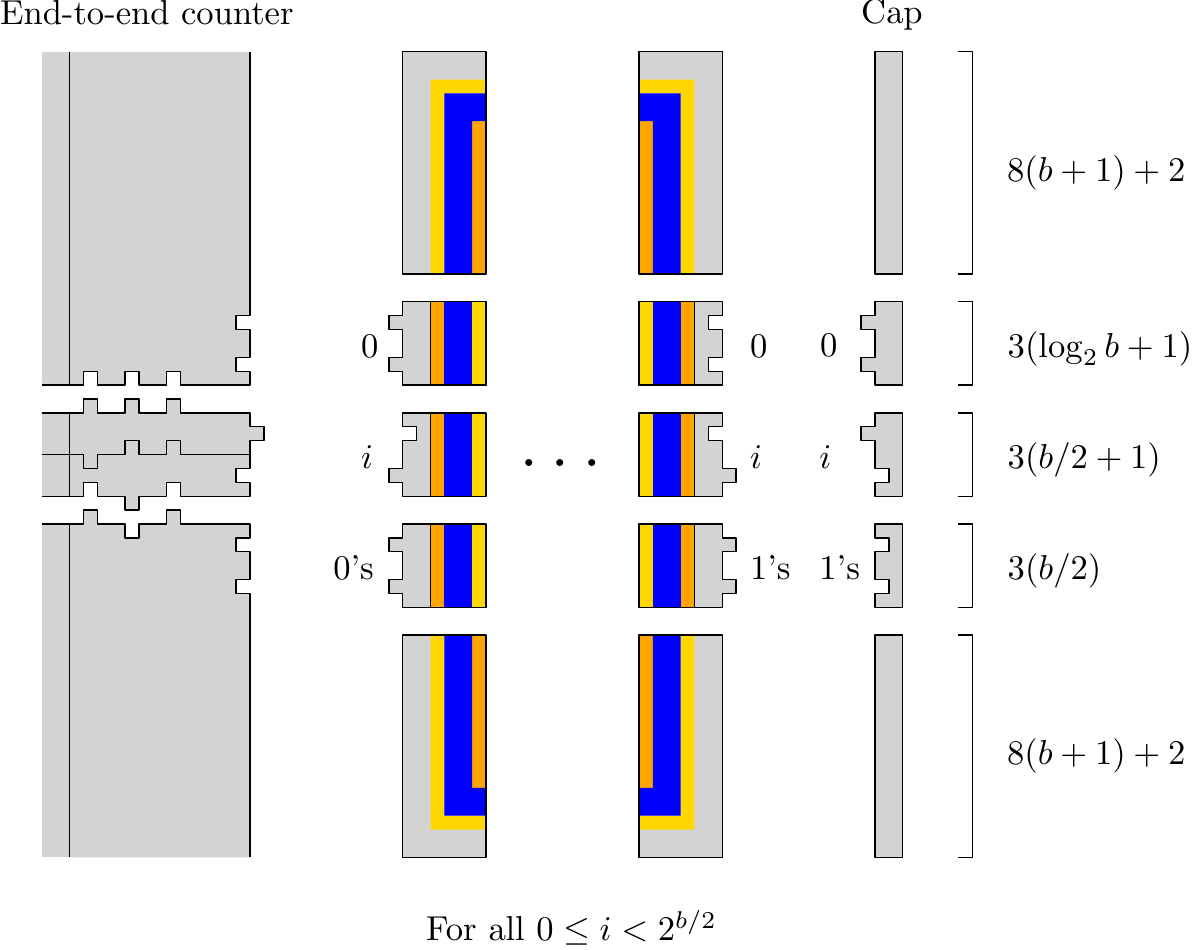}
\caption{(Left) The interaction of a vertical end-to-end counter with the westernmost block in each row. (Right) The cap assemblies built to attach to the easternmost block in each row.}
\label{fig:square-schema-caps}
\end{figure}

After this modification to the east end bar family, mixing all vertical bar families results in $2^{b/2}$ assemblies, each forming most of a row of the block counter.
Mixing these assemblies with the families of horizontal slabs results in a completed set of block counter rows, each containing $2^{b/2}$ square assemblies with dimensions $\Theta(b^2) \times \Theta(b^2)$, forming $2^{b/2}$ rectangles with dimensions $\Theta(2^{b/2}b^2) \times \Theta(b^2)$.

To arrange the rows vertically into a complete block counter, a vertically-oriented version of the end-to-end counter of Section~\ref{sec:sas-much-better-pcfgs-rectangle} with geometry instead of color strips (left part of Fig.~\ref{fig:square-schema-caps}) is assembled and used as a ``backbone'' for the rows to attach into a combined assembly.
This modified end-to-end counter (see Figure~\ref{fig:square-schema-vert-counter}) has subrow values from $0$ to $b/2$, for the $b/2$ most signficant bits of the row value of each block, and row values from $0$ to $2^{b/2}$.
Modified versions of reset bars with height (width in the horizontal end-to-end counter) $\Theta(b^2)$ are used to bridge across the geometry-less portions of the west sides of the blocks, as well as the always-zero $b/2$ least significant bits of the block's row value and subrow $\log_2{b}$ bits.

\begin{figure}[ht]
\centering
\includegraphics[scale=1.0]{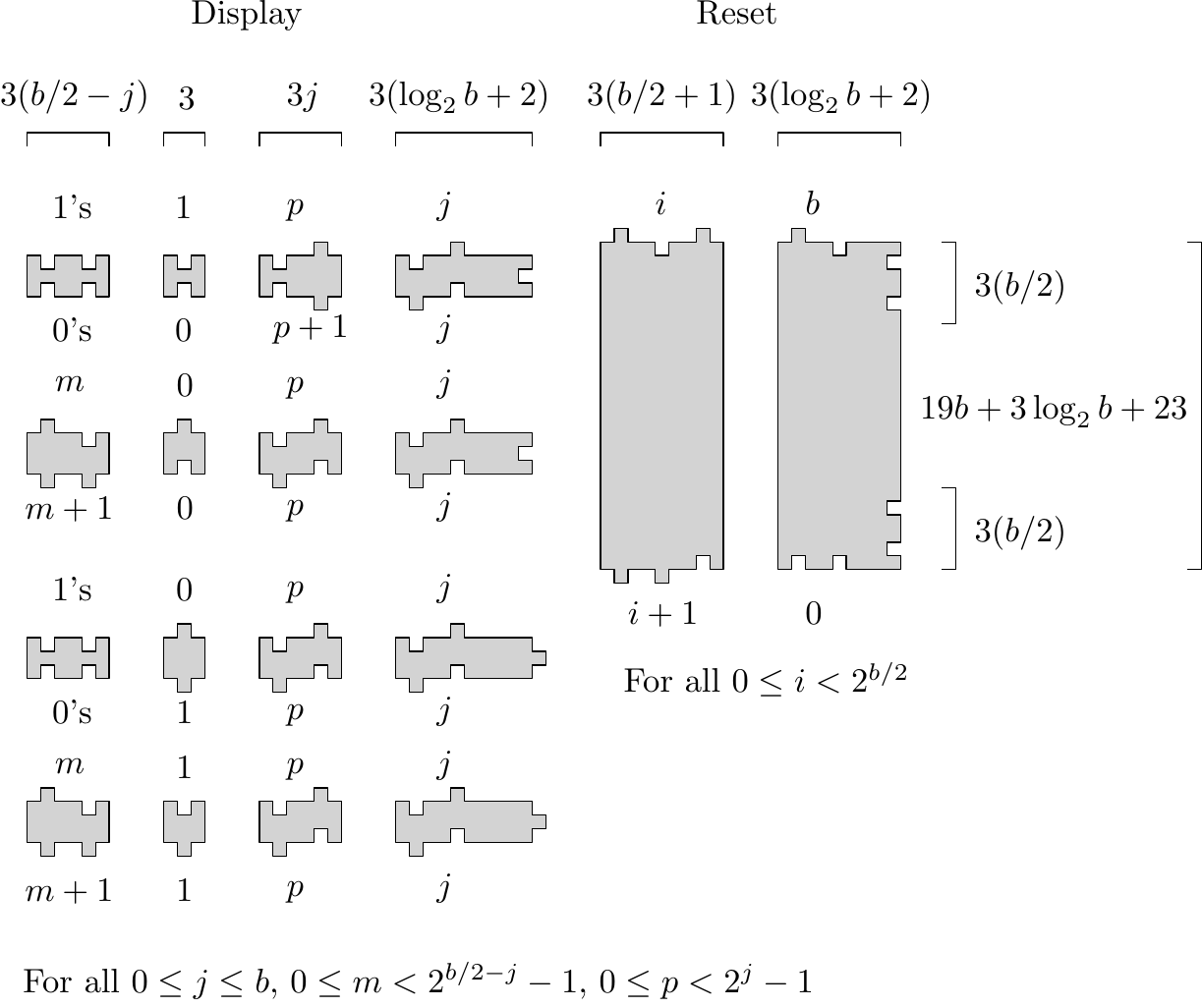}
\caption{The decomposition of the bars of a vertically-oriented end-to-end counter used to combine rows of blocks in a block counter.}
\label{fig:square-schema-vert-counter}
\end{figure}

This modified end-to-end counter can be assembled using $O(b)$ work as done for the original end-to-end counter, since the longer reset bars only add $O(\log(b^2)) = O(\log{b})$ work to the assembly process.
After the vertical end-to-end counter has been combined with the blocks to form a complete block counter, a horizontal end-to-end counter is attached to the top of the assembly to produce a square assembly.
}

\both{
\begin{lemma}
For even $b$, there exists a $\tau=1$ SAS of size $O(b)$ that produces a $b$-bit block counter.
\end{lemma}
}

\later{
\begin{proof}
The construction described builds families of vertical bars and horizontal slabs that are used to assemble each the rings forming all blocks in the counter.
There are a constant number of families, and each family can be assembled using $O(b)$ work.
The vertical and horizontal end-to-end counters can also be assembled using $O(b)$ work each by Lemma~\ref{lem:sas-rect-ub}.
Then the $b$-bit block counter can be assembled by a SAS os size $O(b)$.
\end{proof}

We now consider a lower bound for any PCFG $G$ deriving the counter, using a similar approach as Lemma~\ref{lem:pcfg-end-to-end-counter-lower-bound}.
}

\both{
\begin{lemma}
\label{lem:pcfg-block-counter-lower-bound}
For any PCFG $G$ deriving a $b$-bit block counter, $|G| = \Omega(2^b)$.
\end{lemma}
}

\later{
\begin{proof}

Define a \emph{minimal block spanner} as to be a non-terminal symbol $N$ in $G$ with production rule $N \rightarrow (B, (x_1, y_1)) (C, (x_2, y_2))$ such that the polyomino derived by $N$ (denoted $p_N$) contains a path from a gray cell outside the color loop of the end ring of the counter to a gray cell inside the start color loop of the counter, and the polyominoes derived by $B$ and $C$ (denoted $p_B$ and $p_C$) do not.

First we show that any minimal block spanner is a spanner for at most one block.
Assume by contradiction and that $N$ is a minimal block spanner for two blocks $\block_i$ and $\block_j$ and that $p_B$ contains a gray cell inside the start color loop of $\block_i$.
Then $B$ must be entirely contained in the color loop of the end ring of $\block_i$, as otherwise $N$ is not a minimal block spanner for $\block_i$.
Similarly, $C$ must then be entirely contained in the color loop of the end ring of $\block_j$.
Since no pair of color loops from distinct blocks have adjacent cells, $p_N$ is not a connected polyomino and so $G$ is not a valid PCFG.

\begin{figure}[ht]
\centering
\includegraphics[scale=1.0]{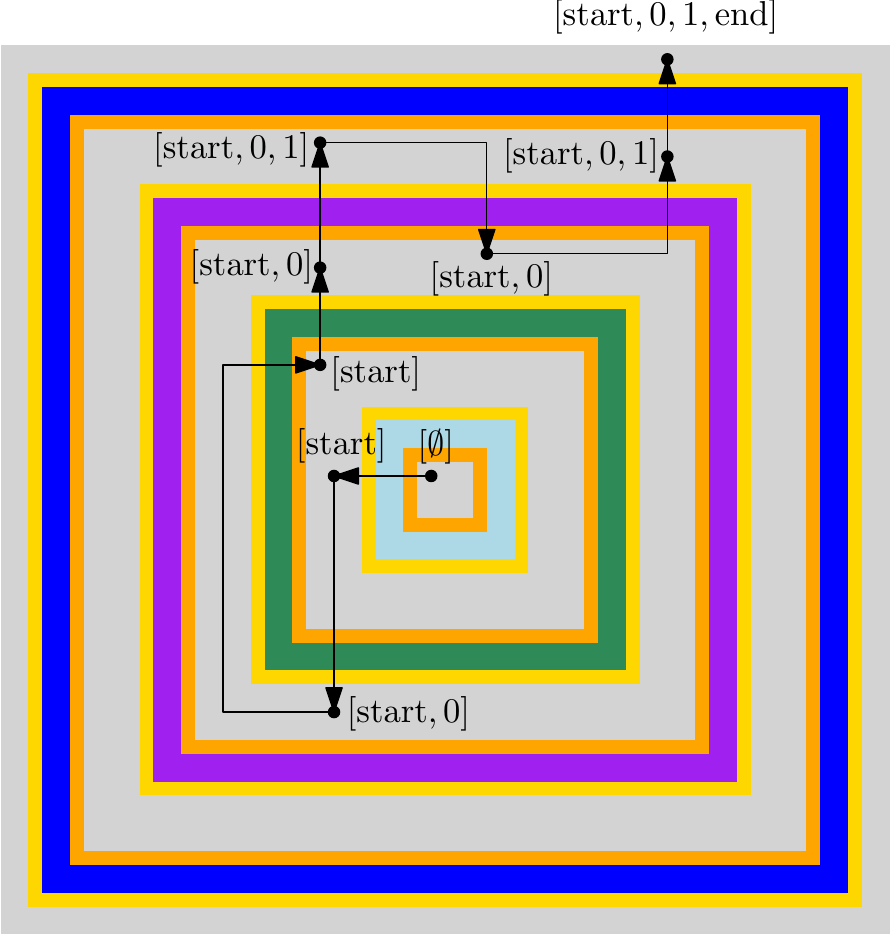}
\caption{A schematic of the proof that the block spanned by a minimal row spanner is unique.
Maintaining a stack while traversing a path from the interior of the start ring to the exterior of the end ring uniquely determines the block spanned by any minimal block spanner containing the path. }
\label{fig:square-proof1}
\end{figure}

Next we show that the block spanned by $N$ is unique, i.e. $N$ cannot be reused as a minimal spanner for multiple blocks.
See Figure~\ref{fig:square-proof1}.
Let $N$ be a minimal spanner for a block $\block_i$ and $p$ be a path of cells in $p_N$ starting at a gray cell contained in the start ring of $\block_i$ and ending at a gray cell outside the end ring of $\block_i$.
Consider a traversal of $p$, maintaining a stack containing the color loops crossed during the traversal.
Crossing a color loop from interior to exterior (a sequence of dark orange, then green, purple, or blue, then light orange cells) adds the center subloop's color to the stack, and traversing from exterior to interior removes the topmost element of the stack.

We claim that the sequence of subloop colors found in the stack after traversing an end ring from interior to exterior encodes a unique sequence of display rings and thus a unique block.
To see why, first consider that the color loop of every ring forms a simple closed curve.
Then the Jordan curve theorem implies that entering or leaving each region of gray cells between adjacent color loops requires traversing the color loop.
Then by induction on the steps of $p$, the stack contains the set of rings \emph{not containing} the current location on $p$ in innermost to outermost order.
So the stack state after exiting the exterior of the end ring uniquely identifies the block containing $p$ and $N$ is a minimal spanner for this unique block.

Since there are $2^b$ distinct blocks in a $b$-bit block counter, any PCFG that generates a counter has at least $2^b$ non-terminal symbols and size $\Omega(2^b)$.
\end{proof}
}

\both{
\begin{theorem}
The separation of PCFGs over $\tau=1$ SASs for constant-label squares is $\Omega(n/\log^3{n})$.
\end{theorem}
}

\later{
\begin{proof}
By construction, a $b$-bit block counter has size $\Theta(2^bb^2) = n$ and so $b = \Theta(\log{n})$.
By the previous two lemmas, the separation is $\Omega((n/b^2)/b) = \Omega(n/\log^3{n})$.
\end{proof}

Unlike the previous rectangle construction, it does not immediately follow that a similar separation holds for 2-label squares.
Finding a construction that achieves nearly-linear separation but only uses two labels remains an open problem.
}

\subsection{Constant-glue constructions}
\label{sec:constant-glues}

Lemma~\ref{lem:constant-glues} proved that any system $\mathcal{S}$ can be converted to a slightly larger system (both in system size and scale) that simulates $\mathcal{S}$.
Applying this lemma to the constructions of Section~\ref{sec:sass-much-better-than-pcfgs} yields identical results for constant-glue systems:

\both{
\begin{theorem}
\label{thm:constant-glues}
All results in Section~\ref{sec:sass-much-better-than-pcfgs} hold for systems with $O(1)$ glues.
\end{theorem}
}

\later{
\begin{proof}
Lemma~\ref{lem:constant-glues} describes how to convert any SAS or SSAS $\mathcal{S} = (T, G, \tau, M)$ into a macrotile version of the system $\mathcal{S}'$ that uses a constant number of glues, has system size $O(\Sigma(T)|T| + |\mathcal{S}|)$, and scale factor $O(\log{|G|})$.
Additionally, the construction achieves matching labels on \emph{all} tiles of each macrotile, including the glue assemblies. 
Because the labels are preserved, the polyominoes produced by each macrotile system $\mathcal{S}'$ simulating an assembly system $\mathcal{S}$ in Section~\ref{sec:sass-much-better-than-pcfgs} preserves the lower bounds for PCFGs (Lemmas~\ref{lem:pcfg-weak-counter-lower-bound},~\ref{lem:pcfg-end-to-end-counter-lower-bound}, and~\ref{lem:pcfg-block-counter-lower-bound}) of each construction.
Moreover, the number of labels in the polyomino is constant and so $|\mathcal{S}'| = O(|T| + |\mathcal{S}|) = O(|\mathcal{S}|)$ and the system size of each construction remains the same.
Finally, the scale of the macrotiles is $O(\log{|G|}) = O(\log{|\mathcal{S}|} = O(\log{b})$, so $n$ is increased by a $O(\log^2{b})$-factor, but since $n$ was already exponential in $b$, it is still the case that $b = \Theta(\log{n})$ and so the separation factors remain unchanged.
\end{proof}
}

\section{Conclusion}

As the results of this work show, efficient staged assembly systems may use a number of techniques including, but not limited to, those described by local combination of subassemblies as captured by PCFGs.
It remains an open problem to understand how the efficient assembly techniques of Section~\ref{sec:pcfgs-never-much-better-sas} and Section~\ref{sec:sass-much-better-than-pcfgs} relate to the general problem of optimally assembling arbitrary shapes.

\section*{Acknowledgements}
We thank Benjamin Hescott and anonymous reviewers for helpful comments and feedback that greatly improved the presentation of the paper.

\bibliographystyle{plain}
\bibliography{nanogrammar2}

\appendix
\magicappendix

\end{document}